\documentclass[1p,number,a4paper]{elsarticle}
\usepackage{stix}
\usepackage[pagewise]{lineno}
\usepackage{amsfonts}
\usepackage{amsmath}
\usepackage{amsthm}
\usepackage{mathrsfs}
\usepackage{enumerate}
\usepackage{graphics}
\usepackage{mathtools}
\usepackage{microtype}
\usepackage[defaultlines=4]{nowidow}
\usepackage{xcolor}
\usepackage[czech, english]{babel}
\usepackage{esvect}
\usepackage{complexity}
\usepackage{comment}
\usepackage[many]{tcolorbox}
\usepackage{thm-restate}
\usepackage[all]{xy}
\usepackage{cleveref}

\theoremstyle{remark}
\newtheorem{claim}{$\rhd$ Claim}
\newcommand{\cqed}{\ensuremath{\lhd}}
\makeatletter

\tcolorboxenvironment{theorem}{
colframe=cyan,interior hidden, breakable,before skip=10pt,after skip=10pt,	outer arc=0pt,
	arc=0pt,
	colback=white,
	rightrule=0pt,
	toprule=0pt,
	top=0pt,
	right=0pt,
	bottom=0pt,
	bottomrule=0pt,
 }
\tcolorboxenvironment{corollary}{
	colframe=blue,interior hidden, breakable,before skip=10pt,after skip=10pt,	outer arc=0pt,
	arc=0pt,
	colback=white,
	rightrule=0pt,
	toprule=0pt,
	top=0pt,
	right=0pt,
	bottom=0pt,
	bottomrule=0pt,
}
\newenvironment{claimproof}{\par
	\pushQED{\cqed}%
	\normalfont \topsep6\p@\@plus6\p@\relax
	\trivlist
	\item\relax
	{\itshape
		Proof of the claim\@addpunct{.}}\hspace\labelsep\ignorespaces
}{%
	\hfill\popQED\endtrivlist\@endpefalse
}
\makeatother
\newcommand{\DG}{\vv{G}^+}
\DeclareTColorBox{answer}{O{orange}O{0cm}}{
	breakable,
	outer arc=0pt,
	arc=0pt,
	colback=white,
	rightrule=0pt,
	toprule=0pt,
	top=0pt,
	right=0pt,
	bottom=0pt,
	bottomrule=0pt,
	colframe=#1,
	enlarge left by=#2,
	width=\linewidth-#2,
}
\DeclareTColorBox{solution}{O{orange}O{0cm}}{
	breakable,
	outer arc=0pt,
	arc=0pt,
	colback=white,
	rightrule=0pt,
	toprule=0pt,
	top=0pt,
	right=0pt,
	bottom=0pt,
	bottomrule=0pt,
	colframe=#1,
	enlarge left by=#2,
	width=\linewidth-#2,
}

\renewcommand{\overline}[1]{\bar{#1}}

\newcommand{\disc}{\mathrm{disc}}
\newcommand{\herdisc}{\mathrm{herdisc}}

\newcommand{\Oof}{\mathcal{O}}

\newcommand{\Cc}{\mathscr{C}}
\newcommand{\Dd}{\mathscr{D}}

\newcommand{\Ff}{\mathcal{F}}

\newcommand{\Ss}{\mathscr{S}}

\newcommand{\N}{\mathbb{N}}

\newcommand{\strM}{\mathbf{M}}

\newcommand{\Mon}{{\mathcal F}_\sigma}
\renewcommand{\phi}{\varphi}

\renewcommand{\FO}{\mathrm{FO}}

\newcommand{\shtm}{\,\widetilde{\triangledown}\,}

\renewcommand{\epsilon}{\varepsilon}

\newcommand{\xfootnote}[1]{\ifcsname ismain\endcsname{}\else\footnote{#1}\fi}

\usepackage[all]{xy}

\DeclareMathOperator{\wcol}{wcol}
\DeclareMathOperator{\WReach}{WReach}

\newcommand{\ad}{\overline{\rm d}}

\newcommand{\ERCagreement}{This paper is part of a project that has received funding from the European Research Council (ERC) under the European Union's Horizon 2020 research and innovation programme (grant agreement No 810115 -- {\sc Dynasnet}), from the German
			Research Foundation (DFG) with grant agreement
			No 444419611, and from the french ANR project HOSIGRA (ANR-17-CE40-0022).\\
			\includegraphics[width=.25\textwidth]{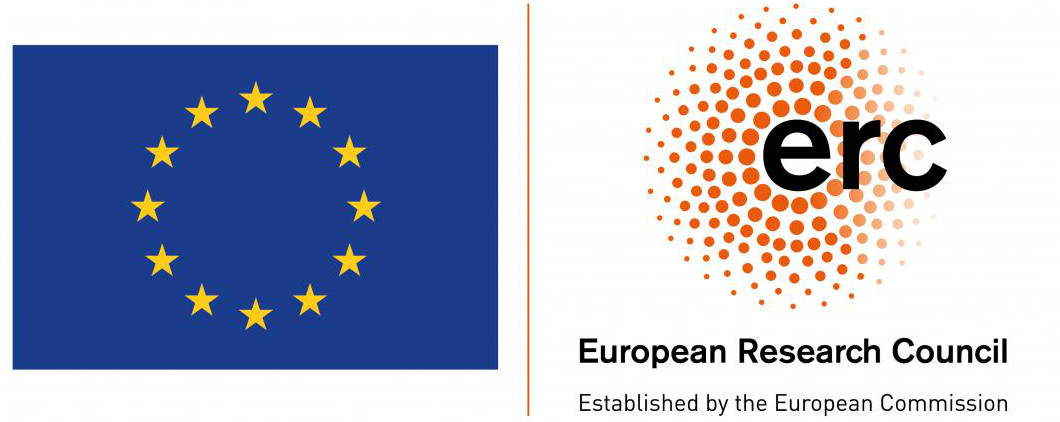}\ 
			\includegraphics[width=.2\textwidth]{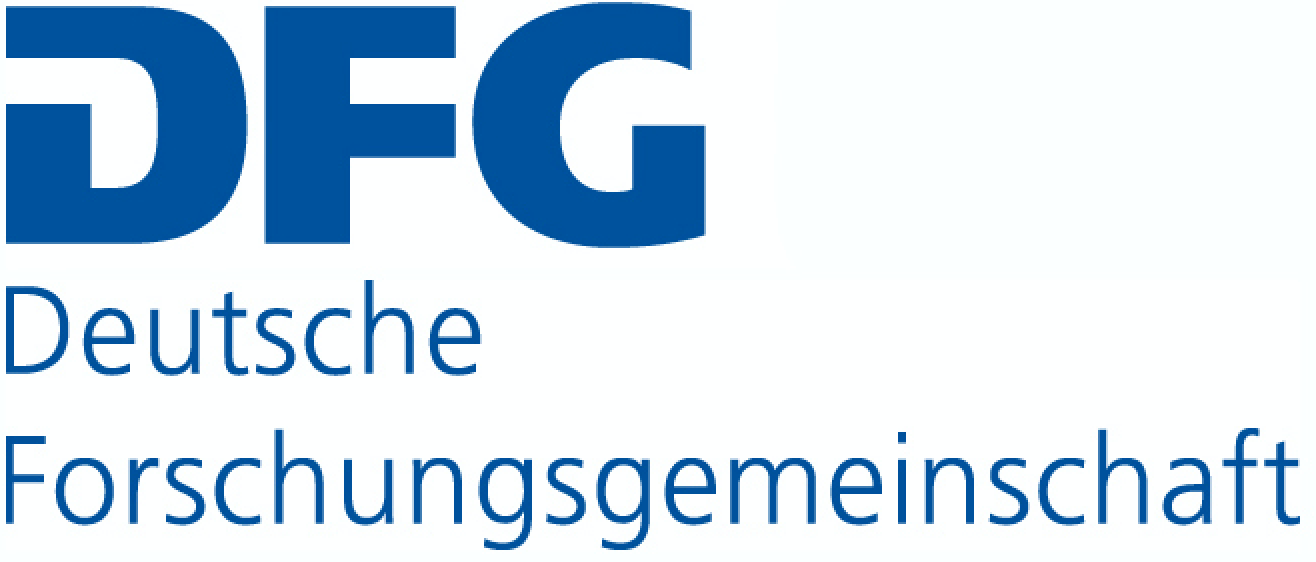}}
	
\newtheorem{theorem}{Theorem}[section]
\newtheorem{ext_theorem}{Theorem}[section]
\newtheorem{corollary}{Corollary}[section]
\newtheorem{definition}{Definition}[section]
\newtheorem{lemma}{Lemma}[section]

\newtheorem{remark}{Remark}[section]
\newtheorem{observation}{Observation}[section]
\newtheorem{problem}{Problem}

\newtheorem{conjecture}[problem]{Conjecture}
\newtheorem{example}{Example}
\crefname{ext_theorem}{Theorem}{Theorems}
\crefname{corollary}{Corollary}{Corollaries}
\crefname{lemma}{Lemma}{Lemmas}
\crefname{section}{Section}{Sections}

\journal{arXiv}
\begin{document}
\begin{frontmatter}
\title{Discrepancy and Sparsity\tnoteref{ERC}}
\tnotetext[ERC]{\ERCagreement}
\author{Mario Grobler}
\ead{grobler@uni-bremen.de}
\address{University of Bremen, Bremen, Germany}
\author{Yiting Jiang}
\ead{yjiang@irif.fr}
\address{Universit\'e de Paris, CNRS, IRIF, F-75006, Paris, France\\and Department of Mathematics, Zhejiang Normal University, China}
\author{Patrice Ossona de Mendez} 
\ead{pom@ehess.fr}
\address{Centre d'Analyse et de Math\'ematiques Sociales (CNRS, UMR 8557),
Paris, France\\
  and     Computer Science Institute of Charles University,
  Praha, Czech Republic}
\author{Sebastian Siebertz}\address{University of Bremen, Bremen, Germany}\ead{siebertz@uni-bremen.de}
\author{Alexandre Vigny}\address{University of Bremen, Bremen, Germany}\ead{vigny@uni-bremen.de}
\begin{abstract}
  We study the connections between the notions of combinatorial
  discrepancy and graph degeneracy.  In particular, we prove that the
  maximum discrepancy over all subgraphs $H$ of
  a graph $G$ of the neighborhood set system of $H$ is sandwiched between
  $\Omega\hspace{0.5pt}(\log\mathrm{deg}(G))$ and
  $\mathcal{O}\hspace{0.5pt}(\mathrm{deg}(G))$, where
  $\mathrm{deg}(G)$ denotes the degeneracy of $G$.  We extend this
  result to inequalities relating weak coloring numbers and
  discrepancy of graph powers and deduce a new characterization of bounded
  expansion classes.

  Then, we switch to a model theoretical point of view, introduce
  pointer structures, and study their relations to graph classes with 
  bounded expansion. We deduce that a monotone class of graphs 
  has bounded expansion if
  and only if all the set systems definable in this class have bounded
  hereditary discrepancy.

  Using known bounds on the VC-density of set systems definable in
  nowhere dense classes we also give a characterization of nowhere
  dense classes in terms of discrepancy.

  As consequences of our results, we obtain a corollary on the
  discrepancy of neighborhood set systems of edge colored graphs, a
  polynomial-time algorithm to compute $\varepsilon$-approximations of
  size $\mathcal{O}\hspace{0.5pt}(1/\varepsilon)$ for set systems
  definable in bounded expansion classes, an application to clique
  coloring, and even the non-existence of a quantifier elimination
  scheme for nowhere dense classes.
\end{abstract}
\end{frontmatter}

\section{Introduction and some motivating examples}\label{sec:intro}

Discrepancy theory emerged from the study of the irregularities of
statistical distributions and number sequences.  It developed and
became a central tool in computational geometry. Two decades ago,
Matou\v sek~\cite{matouvsek1999combinatorial} initiated the study of
\emph{combinatorial discrepancy}, which became a significant subject
in its own right. The combinatorial discrepancy measures the
inevitable irregularities of set systems and the inherent difficulty
to approximate them. (We refer to Section~\ref{sec:prel} for the
formal definitions of the concepts considered in this paper.)

Discrepancy theory offers powerful tools and techniques with many
applications in computational geometry, probabilistic algorithms, derandomization,
communication complexity, searching, machine learning,
pseudo\-randomness, optimization, computer
graphics, and more. Central notions in this theory are also the well
known notions of VC-dimension, $\varepsilon$-nets and
$\varepsilon$-approximations, the latter corresponding to the expected
properties of a pseudorandom set. We refer the reader to the
textbooks~\cite{chazelle2001discrepancy, drmota2006sequences,
  matousek2009geometric} for a presentation of the discrepancy method
and some of its applications.
  
A structural theory of classes of sparse graphs emerged recently,
which is based on the study of densities of shallow minors, generalized
coloring numbers, and constrained orientations~\cite{Sparsity}. In
this setting, two central notions are those of \emph{classes with
  bounded expansion}, which generalize classes excluding a topological
minor, and \emph{nowhere dense classes}, which generalize classes
locally excluding a topological minor. These classes have strong
algorithmic and structural properties. In particular, in a nowhere dense class
it can be checked in almost linear time whether a first-order formula
is satisfied in a given graph from the
class~\cite{DBLP:journals/jacm/GroheKS17}. This last example is only
one among others that witness a strong connection between sparsity
theory and first-order logic.  We refer the interested reader to
\cite{Sparsity} for a comprehensive study of these classes.

The aim of this paper is to establish a bridge between discrepancy
theory and sparsity theory through a study of the combinatorial
discrepancy of set systems defined from sparse classes of graphs by
means of first-order formulas and further to give new
characterizations of degeneracy, bounded expansion, and nowhere denseness in terms of
discrepancy of definable set systems. Before proceeding, we take time
for a few motivating examples.  \medskip

\begin{restatable}{problem}{pbone}
  \label{pb:discdeg}
  Assume that a graph $G$ has the property that  for every
  red/blue coloring of the edges of $G$ there exists a partition
  $(A,B)$ of the vertex set of $G$ such that the number of red (resp.\
  blue) neighbors in $A$ and $B$ of any vertex differ by at most $1$.
  Does $G$ contain a vertex with small degree?
\end{restatable}  

\begin{answer}
It  follows from our results that such a graph $G$ is $1303$-degenerate (see \Cref{sec:app1}).
\end{answer}
\pagebreak
\begin{restatable}{problem}{pbtwo}
  \label{pb:discfo}
  Given a planar graph $G$, find a small subset $F$ of edges such
  that, for every pair~$u,v$ of distinct vertices of $G$, the
  probability that an edge of $G$ belongs to a $uv$-path of length at
  most~$100$ differs from the probability that an edge in $F$ belongs
  to a $uv$-path of length at most~$100$ by at most $\epsilon$.
\end{restatable} 
\begin{answer}
  We prove that a set $F$ of edges of size $\Oof\hspace{0.5pt}(1/\epsilon)$ with the
  prescribed properties can be constructed deterministically in
  polynomial time (see \Cref{sec:app2}).
\end{answer}
\medskip

\begin{restatable}{problem}{pbthree}  
  \label{pb:discK}
  Does there exist a constant $c$, such that the vertices of every map
  graph\xfootnote{A \emph{map graph} is the vertex-face incidence
    graph of a planar map.} $G$ can be colored red or blue, in such a
  way that the difference between the number of red and blue vertices
  in every maximal clique of $G$ is at most $c$?
\end{restatable}
\begin{answer}
  Although there are quite a few reasons to believe that such a
  constant would not exist (it is not even possible in general to
  color the vertices of a perfect graph red and blue in such a way
  that no maximal clique is monochromatic \cite{charbit2016perfect})
  we prove that such a constant~$c$ exists for map graphs (see
  \Cref{sec:app3}).
\end{answer}
\medskip

Last, we also consider the following (seemingly completely unrelated)
problem from sparse finite model theory.  It is known that every class
of finite graphs with bounded expansion has a quantifier elimination
scheme involving unary relations and functions \cite{DKT2}. As it is
known that the fixed-parameter tractability of first-order 
model-checking extends from bounded expansion classes to the more general
nowhere dense classes \cite{Grohe2013,DBLP:journals/jacm/GroheKS17}, it is natural to ask whether
the quantifier elimination scheme also extends. It has been
conjectured that this is not the case, but no proof of this fact was
known. 

\begin{restatable}{problem}{pbfour}
  \label{pb:elim}
  Give an example of a nowhere dense class $\mathscr C$ of
  graphs such that there exists no expansion $\sigma$ of the signature
  of graphs by unary relation and function symbols with the property that
  every first-order formula is equivalent on a $\sigma$-expansion $\mathscr{C}^+$ of 
  $\mathscr C$ to a
  quantifier-free first-order $\sigma$-formula.
\end{restatable}
\begin{answer}
  We prove that the class $\mathscr C$ of $1$-subdivisions of bipartite graphs
  whose girth exceeds the maximum degree has the above
  property. Precisely, there is no expansion $\sigma$ of the signature
  of graphs by unary relation and function symbols in which the formula
  $\phi(x,y)$ expressing that $x$ and $y$ are at distance $2$ in the
  graph is equivalent on a $\sigma$-expansion of 
  $\mathscr C$ to a quantifier-free first-order
  $\sigma$-formula (see \Cref{sec:app4}).
\end{answer} 
\medskip

All these results are obtained as special cases of general theorems
on combinatorial discrepancy of set systems definable in sparse graph
classes, which we prove in this paper.

\pagebreak
This paper is organized as follows: In \Cref{sec:prel} we recall some
needed concepts and notations and present the results obtained in this
paper. In \Cref{sec:deg} we relate the notions of degeneracy and
discrepancy. This connection is extended in \Cref{sec:wcol} to a
relation linking generalized coloring numbers and discrepancy. From
\Cref{sec:ps} onward, we take a (finite) model theoretic point of
view, putting the results obtained in \Cref{sec:deg} in a wider
perspective. In \Cref{sec:pointer-theory}, we translate the
quantifier elimination scheme introduced in \cite{DKT2} for classes
with bounded expansion into a model theoretic language, which allows us
to extend the result obtained in \Cref{sec:wcol} to the general
setting of set systems definable in graphs of a bounded expansion
class. This section is admittedly technical and is basically a formal 
translation of known results. Thus, the reader unfamiliar with model
theory may directly move to \Cref{thm:QE,thm:BE_hdisc}. 
We discuss the case of nowhere dense classes in \Cref{sec:ND}
and applications in \Cref{sec:applications}.

\section{Preliminaries and statement of the results}
\label{sec:prel}

\subsection{Combinatorial discrepancy} 
\label{sec:defdisc}
Let $(U,\Ss)$ be a \emph{set system}, where~$\Ss$ is a collection of
subsets of the \emph{ground set}~$U$.  When the ground set is clear
from the context, we refer to the set system as $\Ss$.  The
\emph{discrepancy} of a mapping $\chi\colon U\rightarrow\{-1,1\}$ on a
set~$S\in\Ss$ is
$\disc_\chi(S)=\bigl|\sum_{v\in S}\chi(v)\bigr|$;
the \emph{discrepancy} of  $\chi$ on~$\Ss$ is the maximum of
$\disc_\chi(S)$ over all $S\in \Ss$, that is, 
$\disc_\chi(\Ss)=\max_{S\in\Ss}\disc_\chi(S)$.
The \emph{(combinatorial) discrepancy} of $\Ss$ is the minimum
discrepancy of a mapping~$\chi\colon U\rightarrow\{-1,1\}$ on~$\Ss$,
that is, 
\[
\disc(\Ss)=\min_{\chi:U\rightarrow\{-1,1\}}\ \max_{S\in\Ss}\ \biggl|\sum_{v\in S}\chi(v)\biggr|.
\]
%
Thus, the combinatorial discrepancy of a set system measures how balanced a
$2$-coloring of this system can be.

It is usual to consider bounds for the discrepancy of a set system
$(U,\Ss)$ in terms of 
$n=\bigl|\bigcup\Ss\bigr|$ (and sometimes in terms of $m=|\Ss|$).  For
instance, by a celebrated result of Spencer~\cite{spencer1985six}, the
discrepancy of a set system 
is 
in $\Oof\hspace{0.5pt}(\sqrt{n\log(m/n)})$, and in the case
where~\mbox{$m=n$} we have $\disc(\Ss)\leq 6\sqrt{n}$. This latter bound is
tight up to the constant. Another important result is the Theorem of
Beck and Fiala~\cite{BECK19811}.
\begin{ext_theorem}[Beck-Fiala Theorem]
	\label{thm:BF}
        The discrepancy of a set system with \emph{degree} at most~$t$
        (that is, each element lies in at most $t$ sets) is less than
        $2t$.
\end{ext_theorem}
This theorem was subsequently improved by Bednarchak and Helm
\cite{bednarchak1997note}, who gave an upper bound of $2d-3$ for
$d\geq 3$. Then Bukh \cite{bukh2016improvement} proved that the upper
bound can be decreased to~$2d-\log^*d$ for sufficiently large $d$.

Many of the techniques to derive upper bounds for the discrepancy are
non-constructive and it is difficult to efficiently find optimal
discrepancy colorings. Given a set system on $n$ elements and
$m\in \Oof\hspace{0.5pt}(n)$ sets, it is \NP-hard to distinguish whether the system
has discrepancy~$0$ or
$\Omega\hspace{0.5pt}(\sqrt{n})$~\cite{charikar2011tight}. In
particular, under the assumption $\P\neq \NP$, one cannot compute a
function $\chi$ whose discrepancy on the set system is within factor
$o\hspace{0.5pt}(\sqrt{n})$ of the discrepancy of the
system. However, in polynomial time one can compute a coloring $\chi$
with discrepancy~$\Oof\hspace{0.5pt}(\sqrt{n})$~\cite{levy2017deterministic}.
Also, the proof of the Beck-Fiala Theorem is constructive, and gives a
polynomial time deterministic algorithm to compute a coloring $\chi$
with discrepancy smaller than twice the degree of the set system.

\pagebreak
A standard example of set systems with high discrepancy is  given by the following.
\begin{example}[Sylvester's example]
	\label{ex:Sylv}
	Sylvester inductively constructed Hadamard matrices $H_{\hspace{-0.7pt}p}$ of order~$2^p$ for every non-negative integer $p$ as follows:
	$H_0=\begin{psmallmatrix*}[r]1\end{psmallmatrix*}$, $H_1=\begin{psmallmatrix*}[r]1&1\\1&-1\end{psmallmatrix*}$, and for $p\geq 1$ 
	\mbox{$H_{\hspace{-0.7pt}p+1}=\begin{psmallmatrix*}[r]H_{p}&H_{p}\\H_{p}&-H_{p}\end{psmallmatrix*}=H_1\otimes H_{\hspace{-0.7pt}p}$}, where $\otimes$ denotes the Kronecker product of matrices.
	Let 
	\[
	\Ss_{\hspace{-0.7pt}p}=\bigl\{\{j\colon (H_{\hspace{-0.7pt}p})_{i,j}=1\}\colon 1\leq i\leq 2^p\bigr\}.
	\]
	Then
	 \[
	 \disc(\Ss_{\hspace{-0.7pt}p})=\Omega\hspace{0.5pt}(\sqrt{n}),
	 \]
	 where $n=|\bigcup \Ss_{\hspace{-0.7pt}p}|$.
\end{example}

Unfortunately, the discrepancy is known to be a \emph{fragile} notion,
as witnessed by the following standard example (see
e.g.~\cite{bansal2014algorithmic}): let $\Ss=\{S_1,\dots,S_m\}$ be a
set system with ground set $U$ and let $\Ss'=\{S_1',\dots,S_m'\}$
be a copy of $\Ss$ with ground set $U'$ disjoint from $U$. Then
the set system $\Ss'' = \{S_1\cup S_1'\dots,S_m\cup S_m'\}$ has always
discrepancy zero, independent of the discrepancy of~$\Ss$.\linebreak
A more robust notion is the \emph{hereditary discrepancy} of a set
system~$(U,\Ss)$, defined as
$	\herdisc(\Ss)=\max_{U'\subseteq U}\disc(\Ss|_{U'})$,
where $\Ss|_{U'}$ denotes the set system $\{S\cap U':
S~\in~\Ss\}$. Moreover, bounding the hereditary discrepancy allows to
bound the sizes of $\varepsilon$-nets and \mbox{$\varepsilon$-approximations}
(what is not the case for discrepancy). For this reason, we focus on
hereditary discrepancy in this paper. Note that the bound in the
Beck-Fiala Theorem applies to the hereditary discrepancy as well, as
considering the trace over a subset of the universe does not increase
the degree of a set system.

\subsection{Sparse graph classes.} 
We consider finite, simple and undirected graphs.  For a graph $G$ we
write~$V(G)$ for its vertex set and $E(G)$ for its edge set. We denote
by $\delta(G)$ the \emph{minimum degree} of $G$, by $\ad(G)$ the
\emph{average degree} of $G$, and by $\omega(G)$ the \emph{clique
  number} of $G$.

We denote by $N_G(v)$ the \emph{neighborhood} of $v$ in $G$, that is,
the set of all the vertices adjacent to~$v$ in $G$. Of prime interest is
the \emph{neighborhood set system} $\Ss^E(G)$ of a graph $G$, which is
defined as $\Ss^E(G)=\{N_G(v): v\in V(G)\}$.

We write $H\subseteq G$ (resp.\ $H\subseteq_i G$) if $H$
is a subgraph (resp.\ an induced subgraph) of $G$. A graph~$G$ is {\em
  $d$-degenerate} if every non-empty induced subgraph of~$G$ has
minimum degree at most~$d$.\linebreak The minimum integer $d$ such that a graph
$G$ is $d$-degenerate is the \emph{degeneracy} $\deg(G)$ of
$G$.\linebreak Hence, $\deg(G)=\max_{H\subseteq_i G}\delta(H)$.

A class $\Cc$ of graphs is called {\em monotone} if it is closed under
taking subgraphs and {\em hereditary} if it is closed under taking
induced subgraphs.  A class~$\mathscr C$ is {\em degenerate} if there
is an integer~$d$ such that all the graphs in $\mathscr C$ are
$d$-degenerate.

The {\em $r$-subdivision} of a graph~$G$ is the graph $G^{(r)}$
obtained by subdividing every edge of~$G$ exactly $r$ times. A
\emph{$\leq$\,$r$-subdivision} of~$G$ is a graph obtained by
subdividing each edge of $G$ at most~$r$ times.  A graph~$H$ is a {\em
  topological minor} of a graph $G$ {\em at depth} $r$ if a 
$\leq$\,$2r$-subdivision of~$H$ is a subgraph of $G$. We denote by
$G\shtm r$ the set of all the topological minors of $G$ at depth~$r$,
and define
$\widetilde{\nabla}_r(G)=\max\{|E(H)|/|V(H)|: H\in G\shtm r\}$.
The two key notions in the theory of sparsity~\cite{Sparsity} are the
notions of \emph{bounded expansion} and \emph{nowhere denseness}, 
defined as follows. 

\pagebreak

\begin{definition} 
  A class $\mathscr C$ of graphs has {\em bounded expansion} if there
  exists a function \mbox{$f\colon\mathbb N\rightarrow\mathbb N$} with
  \begin{equation} \forall G\in\mathscr C\ \ \forall H\in G\shtm
    r\quad\ad(H)\leq f(r).
  \end{equation} 
\end{definition}

\begin{example}
  All proper minor-closed classes of graphs and all classes of graphs
  with bounded maximum degree have bounded expansion.
\end{example}

\begin{definition} 
  A class $\mathscr C$ of graphs is {\em nowhere dense} if there
  exists a function $f\colon\mathbb N\rightarrow\mathbb N$ with
  \begin{equation} \label{eq:nd} \forall G\in\mathscr C\ \ \forall H\in
    G\shtm r\quad\omega(H)\leq f(r).
  \end{equation} 
\end{definition}
Note that every class with bounded expansion is nowhere dense but the
converse does not hold.

\begin{example}
  The class of all graphs $G$ whose girth exceeds the maximum degree
  is nowhere dense but does not have bounded expansion.
\end{example}

Bounded expansion and nowhere dense classes enjoy numerous
characterizations and applications (see \cite{Sparsity}). 
 Among
others, let us mention the characterizations based on \mbox{$p$-centered} colorings~\cite{POMNI} and on (generalized) weak coloring numbers \cite{Zhu2008}.

The \emph{depth} of a rooted forest is the maximum number of vertices in a path linking a root to a leaf. The \emph{treedepth} ${\rm td}(G)$ of a graph $G$ is the minimum height of a rooted forest $Y$ with vertex set $V(Y)=V(G)$ such that adjacent vertices in $G$ belong to a path of $Y$ linking a root to a leaf. For a positive integer $p$, a \emph{$p$-treedepth coloring} of a graph $G$ is a vertex coloring such that every $k\leq p$ color classes induce a subgraph with treedepth at most $k$. Closely related is the notion of a \emph{$p$-centered coloring}, which is a coloring of the vertices such that every connected subgraph with at most $p$ colors on the vertices contains some uniquely colored vertex. It is easily checked that every $(p+1)$-centered coloring is a $p$-treedepth coloring (see  \cite{Taxi_tdepth, Sparsity}). 

\begin{ext_theorem}[Ne\v set\v ril, Ossona de Mendez \cite{POMNI}]
	\label{thm:BEcentered}
	A class $\mathscr C$ of graphs has bounded expansion if and only if there exists a function $f\colon\mathbb N\rightarrow \mathbb N$ such that every graph $G\in\mathscr C$ admits, for each positive integer $p$, a $p$-centered coloring with at most $f(p)$ colors.
\end{ext_theorem}

 We recall the definition of the generalized coloring
numbers, which have been introduced by Kierstead and Yang
\cite{Kierstead2003} as a generalization of the so-called coloring
number.  Let $G$ be a graph and let $L$ be a linear ordering
of~$V(G)$. We say that a vertex $u$ is \emph{weakly $d$-reachable}
from a vertex~$v$ if there exists in $G$ a path~$P$ of length at most
$d$ (possibly $0$) linking $u$ and $v$ such that~$u$ is the minimum
vertex of $P$ with respect to $L$, and we denote by $\WReach_d[G,L,v]$
the set of all vertices weakly $d$-reachable from $v$.  The
\emph{weak coloring number} $\wcol_d(G)$ is defined as the minimum
over all possible linear orderings~$L$ of
$\max_{v\in V(G)}|\WReach_d[G,L,v]|$.

\begin{ext_theorem}[Zhu, \cite{Zhu2008}]
	\label{thm:Zhu}
  A class $\mathscr C$ of graphs has bounded expansion if and only if
  there exists a function $f\colon\mathbb N\rightarrow \mathbb N$
  such that for every graph $G\in\mathscr C$ and every positive
  integer $d$ we have $\wcol_d(G)\leq f(d)$. 
\end{ext_theorem}

Note that nowhere dense classes can also be characterized in terms of
bounds on the weak coloring numbers.

\pagebreak
Structural and algorithmic properties of classes with bounded
expansion and nowhere dense classes have strong links with first-order
logic. In particular, the fixed-parameter linear time first-order
model-checking algorithm for bounded expansion classes \cite{DKT2} is
based on a quantifier elimination scheme, which will be central to the
study conducted in \Cref{sec:pointer-theory}. This also justifies to
extend the study of the discrepancy of neighborhood set systems
presented in \Cref{sec:deg,sec:wcol} to first-order definable set
systems.

\subsection{Model theory and definable set systems.}
In order to establish a bridge between an approach of discrepancy
based on structural graph theory and one based on (finite) model
theory, we consider (from \Cref{sec:ps} onward) set systems defined
from graphs by means of first-order formulas. As a motivating example among others, 
note that first-order logic is the foundation of  query languages 
over relational databases
(see~e.g.~\cite{libkin2013elements}). In this original setting, queries
are formulas and query answers form definable set systems. 

For simplicity, we focus on (vertex
colored) graphs, and on structures with only unary predicates and
unary functions (which we call \emph{pointer structures}). However,
all definitions extend to general structures in a straightforward
way. We use standard notation from graph theory~\cite{diestel} and
finite model theory~\cite{libkin2013elements}, and we refer to the
referenced textbooks for all undefined notation.

Recall that a \emph{signature} $\sigma$ is a set of relation and function symbols with an attached arity. A \emph{$\sigma$-structure} $\mathbf M$ is defined by its \emph{domain} $M$ (which is a set), and an interpretation of the symbols in $\sigma$ as actual relations and functions on $M$. 
The \emph{Gaifman graph} of a structure $\mathbf M$ is the graph with vertex set $M$, where two (distinct) vertices $u,v\in M$ are adjacent if they belong together to a tuple of a relation ($k$-ary functions being interpreted as $(k+1)$-ary relations).

The \emph{quantifier rank} of a first-order formula $\phi$ is the maximum number of nested quantifiers in~$\phi$. 
For example, the quantifier rank of the formula 
$\exists x\,\bigl((\exists y\,(P(x,y)\vee (\forall z\, R(x,yz)))\bigr)\wedge(\forall t\, Z(x,t))$ is three.
We say that two $\sigma$-structures $\mathbf M$ and $\mathbf M'$ are \emph{$q$-equivalent} if they satisfy the same sentences with quantifier rank at most $q$.

Let $G$ be a graph with adjacency relation $E(x,y)$. A set
$S\subseteq V(G)^d$ is \emph{definable} in $G$ \emph{with parameters}
if there exists a partitioned formula $\phi(\bar x;\bar y)$ (with
$\bar x=(x_1,\dots,x_d)$ and \mbox{$\bar y=(y_1,\dots,y_\ell)$}) and an
$\ell$-tuple $\bar b=(b_1,\dots,b_\ell)$ of vertices of $G$ such that
a $d$-tuple $\bar v=(v_1,\dots,v_d)$ belongs to~$S$ if and only if
$G\models\phi(\bar v,\bar b)$. For $\bar b\in V(G)^\ell$, let
$\phi(G,\bar b)=\{\bar v\in V(G)^d: G\models \phi(\bar v,\bar b)\}$ be
the subset of $d$-tuples of vertices of $G$ defined by $\phi$ for the
parameter $\bar b$. The \emph{set system defined by~$\phi(\bar x;\bar y)$} on~$G$ is the set system
$(V(G)^{|\bar x|},\Ss^\phi(G))$, where
$\Ss^\phi(G)=\{\phi(G,\bar b):\bar b\in V(G)^{|\bar y|}\}$. In
particular, the neighborhood set system $\Ss^E(G)$ of a graph $G$ is
defined by the partitioned formula~$E(x;y)$.

\begin{example}[Sylvester's graphs]
	\label{ex:Sp}
	Let $S_{\hspace{-0.7pt}p}$ be the bipartite graph with $2^{p+1}$ vertices, whose adjacency matrix is obtained from the Hadamard matrix $H_{\hspace{-0.7pt}p}$ (described in \Cref{ex:Sylv}) by replacing each entry $-1$ by $0$. Then we have
\[
	\disc(\Ss^E(S_{\hspace{-0.7pt}p}))=\Omega\hspace{0.5pt}\biggl(\bigr|\bigcup \Ss^E(S_{\hspace{-0.7pt}p})\bigr|^{1/2}\biggr),
	\]
	since $\Ss^E(S_{\hspace{-0.7pt}p})$ is the set system obtained from the set system $\Ss_{\hspace{-0.7pt}p}$ introduced in \Cref{ex:Sylv} by duplicating every set (as $H_{\hspace{-0.7pt}p}$ is symmetric).
\end{example}
\pagebreak
Note that we allow formulas without parameter variables. In such a
case, the set system defined by a formula~$\rho(\bar x)$ on $G$ is the
singleton $\{\rho(G)\}$, where
\mbox{$\rho(G)=\{\bar v\in V^{|\bar x|}: G\models \rho(\bar
  v)\}$}. More generally, for a finite set $\Phi$ of formulas
$\phi_i(\bar x,\bar y)$ (with same $|\bar x|$ but possibly
dif\-ferent~$|\bar y|$), the set system defined by $\Phi$ on $G$ is
the set system $(V(G)^{|\bar x|},\Ss^\Phi(G))$, where
\mbox{$\Ss^\Phi(G)=\bigcup_i \Ss^{\phi_i}(G)$}.

\subsection{VC-dimension and shatter functions.}
The Vapnik–Chervonenkis dimension (VC-dimension) is a classical
measure for the complexity of set systems~\cite{Vapnik1971}.  We say
that a subset $X$ of the ground set $U$ of a set system $(U,\Ss)$ is
\emph{shattered} by $\Ss$ if~$\Ss|_X$ contains all the subsets
of~$X$. The \emph{VC-dimension} of $\Ss$ is the maximum size of a
subset shattered by $\Ss$.  Classes $\Cc$ of graphs with the property
that for all first-order formulas $\phi$, there exists a constant $d$
such that the system $\Ss^\phi(G)$ has VC-dimension bounded by~$d$ for
all $G\in \Cc$ play a key role in model theory: they are called
\emph{dependent classes}~\cite{shelah1990classification}.
Classes whose monotone closure is dependent are exactly nowhere dense
classes \cite{Adler2013}.

A more precise description of the complexity of a set system $\Ss$ is
given by its \emph{(primal) shatter function}~$\pi_\Ss$, which is defined as
follows: for an integer $m$, we define $\pi_\Ss(m)$ to be the maximum
over all subsets $X\subseteq U$ of size $m$ of the number of distinct
sets in the collection $\Ss|_X$.  By the famous Sauer-Shelah
Lemma~\cite{sauer1972density,shelah1972combinatorial,Vapnik1971} we
have $\pi_\Ss(k)\in\Oof\hspace{0.5pt}(m^d)$ if the VC-dimension of $\Ss$ is at
most~$d$.


A strong connection between these notions and combinatorial
discrepancy is given by the next fundamental result, which is known to
be optimal~\cite{alon1999norm}.
  
\begin{ext_theorem}[Matou\v{s}ek~\cite{matouvsek1995tight} and Matou\v{s}ek, Welzl, and Wernisch~\cite{matouvsek1993discrepancy}]
\label{thm:matousek}
Let $\Ss$ be a set system, let $n$ be the size of the ground set of
$\Ss$, and let $d,C$ be constants, such that $\pi_\Ss(m)\leq Cm^d$ for
all~$m\leq n$. Then the discrepancy of $\Ss$ is bounded
by 
\begin{equation*}
  \disc(\Ss)= \Oof\bigl(n^{1/2-1/2d}\bigr), \text{ if }d>1,\text{ and } \disc(\Ss)=\Oof\bigl(\log^{5/2}n\bigr),\text{ if }d=1. 
\end{equation*}
\end{ext_theorem}

Each set system $\Ss$ with ground set $U$ can be represented by its \emph{incidence graph} $I_\Ss$ with parts~$U$ and $\Ss$, where $u\in U$ is adjacent to $S\in\Ss$ if $u\in S$.
The \emph{dual} of a set system $\Ss$ with ground set $U$ is the set system 
$\Ss^\ast$, whose incidence graph is obtained from the incidence graph of $\Ss$ by exchanging the parts. (Note that this corresponds to the notion of hypergraph duality.) 

The \emph{dual shatter function} $\pi_\Ss^\ast(m)$ of a set system $\Ss$ is defined as the primal shatter function of the dual set system $\Ss^\ast$.
The discrepancy of a set system can also be bounded in terms of the dual shatter function (where the bound is also tight \cite{matouvsek1995tight}). As noticed by Matou\v sek~\cite{matouvsek1995tight}, the next theorem follows from a slight modification of  the bound obtained by Matou\v{s}ek, Welzl, and Wernisch~\cite{matouvsek1993discrepancy} for the discrepancy of set systems defined by half-spaces, by using (in the abstract setting) the packing lemma due to  Haussler \cite{Ha}, instead of the elementary lemma concerning the volume of an $r$-ball in an arrangement of hyperplanes due to Chazelle and Welzl.
\begin{ext_theorem}
	\label{thm:MWW}
	Let $\Ss$ be a set system, let $n$ be the size of the ground set of
	$\Ss$, and let $d,C$ be constants, such that $\pi_\Ss^*(m)\leq Cm^d$ for
	all $m\leq n$. Then the discrepancy of $\Ss$ is bounded
	by
\begin{equation*}
	 \disc(\Ss)=\Oof\bigl(n^{1/2-1/2d}\,\sqrt{\log n}\bigr).
\end{equation*}
\end{ext_theorem}

%

 \pagebreak
Bounded expansion and nowhere dense classes 
can  be  characterized by shatter functions:
\begin{ext_theorem}[Pilipczuk, Siebertz,
  Toru\'nczyk~\cite{pilipczuk2018number}]
\label{thm:types}
Let $\Cc$ be a monotone class of graphs.
Then
\begin{itemize}
\item the class $\Cc$ has bounded expansion if and only if for every
  formula~$\phi(\bar x;\bar y)$ there exists a constant $C$ with
  $\pi_{\Ss^\phi(G)}(m)\leq C m^{|\bar y|}$ and
  $|\bigcup \Ss^\phi(G)|_A|\leq C|A|^{|\bar y|}$ for every $G\in\Cc$
  and every $A\subseteq V(G)$;
\item the class $\Cc$ is nowhere dense if and only if for every
  formula~$\phi(\bar x;\bar y)$ and for every $\epsilon>0$ there
  exists a constant $C$ with
  $\pi_{\Ss^\phi(G)}(m)\leq C m^{|\bar y|+\epsilon}$ and
  $|\bigcup \Ss^\phi(G)|_A|\leq C|A|^{|\bar y|+\epsilon}$ for every
  $G\in\Cc$ and every $A\subseteq V(G)$.
\end{itemize}
\end{ext_theorem}

From this result, one can derive a bound on the hereditary discrepancy
of sets systems definable in nowhere dense and bounded expansion
classes. In the case of a nowhere dense class, the derived bound is
established in \Cref{lem:disc-nd}; in the bounded expansion case, our
result will establish a constant upper bound, clearly improving the
bound derived from \Cref{thm:matousek,thm:types}.

\subsection{Our results.}

In \Cref{sec:deg}, we prove that the notions of discrepancy and degeneracy are deeply linked. Precisely, we prove that the maximum discrepancy over all subgraphs $H$ of a graph $G$ of the neighborhood set system of $H$ is sandwiched between $\Omega\hspace{0.5pt}(\log\deg(G))$ and $\Oof\hspace{0.5pt}(\deg(G))$ (see \Cref{thm:discdeg}). In \Cref{sec:wcol}, we extend this result to
 inequalities relating weak coloring numbers and discrepancy of graph powers.
Precisely,
the  maximum hereditary discrepancy over all subgraphs~$H$ of the power $G^d$ of a graph
$G$ of the neighborhood set system of $H$ is sandwiched between $\Omega\hspace{0.5pt}(\log \wcol_{\lceil d/2\rceil}(G))$ and $\Oof\hspace{0.5pt}(\wcol_d(G)^2)$ (\Cref{thm:wcol}), and we deduce a  first characterization of bounded expansion classes in terms of discrepancy 
(\Cref{cor:BEpower}).

In order to extend these results further, we switch to a model theoretic point of view.
In \Cref{sec:ps}, we introduce pointer structures, which are structures with only unary relations and unary functions, and prove that set systems definable by a quantifier-free formula in such structures have bounded hereditary discrepancy (\Cref{thm:QF}). As a corollary, we deduce that every set system (induced by a finite ground set) definable in an infinite set has bounded discrepancy (\Cref{cor:RPS}) and that the neighborhood set systems of the graphs in a set-defined class have bounded hereditary discrepancy (\Cref{thm:setdef}).
Then, we introduce the first-order theories  ${\rm Th}^{\rm TF}_{\sigma,\rho}$ for pointer structures, mimicking the characteristic properties of bounded expansion classes and prove that they have quantifier elimination (\Cref{thm:QE}). From this, we deduce a characterization of bounded expansion classes as the monotone classes whose definable set systems have bounded hereditary discrepancy (\Cref{thm:BE_hdisc}).

Then, using bounds on the VC-density of set systems definable in nowhere dense classes proved in~\cite{pilipczuk2018number}, we give a characterization of nowhere dense classes in terms of discrepancy (\Cref{thm:ND}). We believe that our upper bounds on discrepancy can be improved in this case, and we propose a conjecture for the optimal bound (\Cref{conj:ND}).

In \cref{sec:applications}, we provide some corollaries on edge colored graphs (\cref{cor:edgecol}), \linebreak\mbox{$\epsilon$-approximations} (\Cref{cor:approx}), clique coloring (\cref{thm:discK}), and quantifier elimination schemes (\cref{thm:noQE}), which allows us to solve the motivating problems presented in the introduction.

Finally, in \Cref{sec:conc}, we discuss some possible extensions of this work.

\def\ismain{yes!}

\section{Discrepancy and Degeneracy}
\label{sec:deg}
In this section, we relate the degeneracy of a graph $G$ to the discrepancy of
the neighborhood set system $\Ss^E(G)$ of $G$.  Our main
result is the following theorem, which directly follows from
\Cref{lem:deg_lower,lem:deg_upper} proved below.

\begin{theorem}
  \label{thm:discdeg}
  For every graph $G$ we have
  \begin{equation}
    \frac{\log_2(\pi\deg(G))}{4}-2\leq \max_{H\subseteq G}\disc(\mathscr S^E(H))<3\deg(G).
  \end{equation}
\end{theorem}


\paragraph{About the upper bound}
The bipartite graph $G$ with a part of size $n$ and a part of size
$\binom{n}{d}$ corresponding to all possible neighborhoods of size $d$
in the part of size $n$ is such that
$\max_{H\subseteq G}\disc(\mathscr S^E(H))\geq d=\deg(G)$. Hence, the
upper bound of \Cref{thm:discdeg} is tight up to a constant
factor.

\paragraph{About the lower bound}
As Spencer's bound $\disc(\Ss)\leq 6\sqrt{n}$ for every set system
$\Ss$ on a ground set of size $n$ and $\Oof\hspace{0.5pt}(n)$ many sets is known to
be tight up to a constant factor (see \Cref{ex:Sp}), we have
$\max_{H\subseteq K_n}\disc(\mathscr
S^E(H))\in\Omega\hspace{0.5pt}(\deg(K_n)^{1/2})$.  It is possible that the lower
bound given in \Cref{thm:discdeg} might be improved to
$\Omega\hspace{0.5pt}(\deg(G)^c)$ for some positive constant $c$. 

\paragraph{About the maximum over subgraphs}
In the statement, we consider $\max_{H\subseteq G}\disc(\Ss^E(H))$,
which is a monotone form of the discrepancy. Instead of
$\disc(\Ss^E(H))$ we could have considered
$\herdisc(\Ss^E(H))$. However, this would have led to the same result,
as for every graph $G$ we have
\[
\max_{H\subseteq G}\herdisc(\Ss^E(H))=\max_{H\subseteq G}\disc(\Ss^E(H)).
\]
Indeed, as $\herdisc(\Ss^E(H))\geq \disc(\Ss^E(H))$ we get that
$\max_{H\subseteq G}\herdisc(\Ss^E(H))$ is at least
$\max_{H\subseteq G}\disc(\Ss^E(H))$.  For the reverse inequality, it
will be sufficient to prove
$\herdisc(\Ss^E(G))\leq \max_{H\subseteq G}\disc(\Ss^E(H))$.  Let
$V'\subseteq V(G)$ be such that
$\herdisc(\Ss^E(G))=\disc(\Ss^E(G)|_{V'})$, let $W\subseteq V(G)$ be
the set of all neighbors in $G$ of some vertex in $V'$, and let~$H$ be
the subgraph of~$G$ with vertex set $V'\cup W$, where we keep all the
edges of $G$ with some endpoint in $V'$.  Then
$\Ss^E(G)|_{V'}\subseteq\Ss^E(H)$, thus
$\herdisc(\Ss^E(G))\leq \disc(\Ss^E(H))$.  \bigskip





We now prove that the discrepancy of the neighborhood set system is
linearly bounded by the degeneracy.
\begin{lemma}
  \label{lem:deg_upper}
  Let $G$ be a graph. Then
  $\herdisc(\Ss^E(G)) < 3\deg(G)$.
\end{lemma}
\begin{proof}
  Let $d=\deg(G)$ and let $\vec G$ be an orientation of $G$ with
  maximum out-degree $d$. Let~$N^-(v)$ and $N^+(v)$ denote,
  respectively, the in-neighborhood and the out-neighborhood of a
  vertex~$v$. Consider the set system $\Ss_1=\{N^-(v):v\in
  V(G)\}$. Every vertex $v$ in~$V(G)$ belongs to at most $d$ sets in
  $\Ss_1$, for if $v\in N^-(u)$, then $u\in N^+(v)$.  According to the
  Beck-Fiala Theorem, for every subset $X$ of vertices there exists a
  function $\chi\colon X\to \{-1,1\}$ such that for every vertex~$v$
  we have $\Bigl|\sum_{u\in N^-(v)\cap X}\chi(u)\Bigr|< 2d$. As
  $|N^+(v)\cap X|\leq d$ we have
  $\Bigl|\sum_{u\in N_G(v)\cap X}\chi(u)\Bigr|< 3d$.  Hence we have
  $\disc(\Ss^E(G)_{|X}) < 3\deg(G)$.
\end{proof}

\begin{remark}
  Using Bukh's improvement \cite{bukh2016improvement}, if $\deg(G)$ is
  sufficiently large, then the bound can
  be decreased to $3\deg(G)-\log^*\deg(G)$. 
\end{remark}

\begin{lemma}
  \label{lem:deg_lower}
  Let $G$ be a graph and let $c=\max_{H\subseteq
    G}\disc(\Ss^E(H))$. Then $G$ is $\frac{16^{c+2}}{\pi}$-degenerate.
\end{lemma}

\begin{proof}
  In the following we consider $G$ as a vertex-labeled graph and 
  all subgraphs as vertex-labeled graphs (that is, we do not identify
  isomorphic subgraphs). 
  First assume that $G$ is bipartite.  Let $A$ and $B$ be the two
  parts of the bipartition with \mbox{$|A|\geq |B|$}.  For every
  subgraph $H\subseteq G$ denote by $\Ss_H$ the set system
  \mbox{$\{N_H(v) : v\in A\cap V(H)\}$}. Note that $\Ss_H\subseteq \Ss^E(H)$ and
  $\bigcup\Ss_H\subseteq B$.  Hence $\disc(\Ss_H)\leq c$ with witness
  coloring \mbox{$\gamma_H\colon B\rightarrow \{-1,1\}$}.
	
  Let $d=\deg(G)$, $n=|B|$, and $m=|E(G)|$. The graph $G$ has $2^m$
  spanning subgraphs and~$2^n$ ways to color~$B$ (with colors in
  $\{-1,1\}$). Hence, there exists a coloring $\gamma$ such that the
  set~$\Ff_\gamma$ of all spanning subgraphs $H$ of $G$ with
  $\disc_\gamma(\Ss_H)\leq c$ has size at least $2^m/2^n$.
	
	\smallskip
  For $v\in A$ and $\alpha\in\{-1,1\}$, define
  $N_G^{\alpha}(v)=N_G(v)\mathop{\cap}\gamma^{-1}(\alpha)$.  Then 
  for every spanning subgraph $H$ of $G$ the
  inequality \mbox{$\disc_\gamma(\Ss_H)\leq c$} rewrites as
  \[\forall v\in A\quad\bigl||N_H^{-1}(v)|- |N_H^1(v)|\bigr|\leq c.\]
  As $G$ is bipartite, these conditions on the neighborhoods of the
  vertices in $A$ are independent and the number of graphs $H$ in
  $\Ff_\gamma$ is the product over all vertices $v\in A$ of the number of pairs
  $(X_v,Y_v)\subseteq N_G^{-1}(v)\times N_G^1(v)$ with
  $\bigl||X_v|-|Y_y|\bigr|\leq c$.
	
  Let $d^{-1}(v)=|N_G^{-1}(v)|$ and $d^{1}(v)=|N_G^1(v)|$. By
  considering supersets of $N_G^{-1}(v)$ and $N_G^1(v)$ obtained by
  adding $c$ dummy elements, one easily checks that the number of
  pairs $(X_v,Y_v)\subseteq N_G^{-1}(v)\times N_G^1(v)$ with
  $\bigl||X_v|-|Y_v|\bigr|\leq c$ is bounded by
  $\sum_{k=0}^{m(v)}\binom{d^{-1}(v)+c}{k}\,\binom{d^1(v)+c}{k}$, where $m(v)=\min(d^{-1}(v),
  d^1(v))+c$. 
 Thus,
  \[
    |\Ff_\gamma|\leq\prod_{v\in A}\sum_{k=0}^{m(v)}\binom{d^{-1}(v)+c}{k}\, \binom{d^1(v)+c}{k}.\]
	
  \begin{claim}
    Let $a\leq b$ be positive integers and let
    $s=\lceil (a+b)/2\rceil$. Then
    \[\sum_{k=0}^{a}\binom{a}{k}\,
      \binom{b}{k}\leq \binom{2s}{s}.\]
  \end{claim}
  \begin{claimproof}
    By replacing $b$ by $b+1$ if necessary, we can assume that $a+b$
    is even.  Assume first that $a<b$, hence, $a+1\leq b-1$. Then for every
    $0\leq k\leq a$ we have
    $\frac{a+1}{a+1-k}\cdot \frac{b-k}{b}\geq 1$.  Thus,
    \begin{equation*}
      \binom{a+1}{k}\,\binom{b-1}{k}=\frac{a+1}{a+1-k}\,\binom{a}{k}\,\frac{b-k}{b}\,\binom{b}{k}\geq\binom{a}{k}\,\binom{b}{k}.
    \end{equation*}
    It follows that we can reduce to the case $a=b$.
    Then
    $\sum_{k=0}^a\binom{a}{k}\,\binom{b}{k}=\sum_{k=0}^{s}\binom{s}{k}^2=\binom{2s}{s}$,
    which is the special case $m=n=r=s$ of the Chu-Vandermonde
    identity
    $\sum_{k=0}^{r}\binom{m}{k}\,\binom{n}{r-k}=\binom{m+n}{r}$.
  \end{claimproof}
	
	\pagebreak
  \def\dummy{\cite[Lemma 17.5.1]{Cover2006}} Now we use the following
  upper bound.
  \begin{claim}[\dummy]
    For $0<k<n$ we have
    \[\binom{n}{k}\leq \sqrt{\frac{n}{\pi k
          (n-k)}}2^{nH(k/n)},\] where
    $H(p)=-p\log_2(p)-(1-p)\log_2(1-p)$, which is the standard entropy
    function.  \hfill\cqed
  \end{claim}
	
  As $H(1/2)=1$ we have
  $\binom{2s}{s} \leq \frac{4^{s}}{\sqrt{\pi s/2}}$.  Let $h(v)=\lceil \frac{d(v)}{2}+c\rceil\geq \delta(G)/2$.
We have
	
  \begin{align*}
    |\Ff_\gamma|&\leq\prod_{v\in A}
                  \sum_{k=0}^{m(v)}\binom{d^{-1}(v)+c}{k}\,
                  \binom{d^1(v)+c}{k} \\
		&\leq \prod_{v\in A}\binom{2h(v)}{h(v)} 
           \leq \prod_{v\in A}\frac{2^{2h(v)}}{\sqrt{\pi h(v)/2}}\\
		&\leq \prod_{v\in A}\frac{2^{d(v)+2c+1}}{\sqrt{\pi\delta(G)/4}}
           \leq\prod_{v\in A}\,\Bigl( \frac{4^{c+1}}{\sqrt{\pi\delta(G)}}\,2^{d(v)}\Bigr)\\
		&\leq	\Bigl( \frac{4^{c+1}}{\sqrt{\pi \delta(G)}}\Bigr)^{|A|} 2^m.
  \end{align*}
	
  As $|\Ff_\gamma|\geq\left(\frac{1}{2}\right)^{n} 2^m$ we get 
  $\frac{4^{c+1}}{\sqrt{\pi \delta(G)}}\geq 2^{-\frac{n}{|A|}}$.
  Hence, as
 $n=|B|\leq |A|$, we get
  $\delta(G) \leq \frac{64}{\pi}\cdot 16^c$. As this holds for every
  induced subgraph, we get $\deg(G) \leq \frac{64}{\pi}\cdot 16^c$.
	
  Now, if $G$ is not bipartite, we get that the degeneracy is at most
  $\frac{16^{c+2}}{\pi}$, as $G$ includes a bipartite subgraph whose
  degeneracy is at least one fourth of the degeneracy of $G$ (first
  construct a bipartite subgraph $G'$ with at least half of the edges of $G$, 
  then find a subgraph whose minimum degree is at least the edge
  density of $G'$, see e.g.~\cite{Sparsity}).
\end{proof}

%

\section{Discrepancy and Generalized Coloring Numbers}
\label{sec:wcol}

In this section, we extend the result of \Cref{sec:deg} by considering the (generalized) weak coloring numbers.
%
%
%
%
%
 Our main result is the following.

\begin{theorem}
\label{thm:wcol}
	Let $G$ be a graph and let $d$ be a positive integer. Then
\[
\begin{split}
\frac{\log_2({\rm wcol}_{\lceil d/2\rceil}(G))}{6(d+1)}-\frac{\log_2 (d+1)}{3}-\frac32\leq &\max_{d'\leq d}\max_{H\subseteq G}\herdisc(\Ss^E(H^{d'}))\\
&\qquad\qquad < (2d\wcol_{d-1}(G)+1)\wcol_d(G).
\end{split}
\]
\end{theorem}

Note that, as $\wcol_1(G)=\deg(G)+1$,  \Cref{thm:wcol} yields (for $d=1$)
	\[
	\frac{\log_2(\deg(G)+1)}{12}-\frac{11}{6}\leq \max_{H\subseteq G}\herdisc(\Ss^E(H))< 3\deg(G)+3,
	\]
to be compared with \Cref{thm:discdeg}.

\pagebreak
As a corollary of \Cref{thm:wcol} we obtain the following characterization of classes with bounded expansion.

\begin{corollary}
	\label{cor:BEpower}
	Let $\mathscr C$ be a monotone class of graphs. 
	Then $\mathscr C$ has bounded expansion if and only if 
	the hereditary discrepancy of $\mathscr S^E(G^k)$ is bounded on $\mathscr C$ for each positive integer $k$.
\end{corollary}

%

The proof of  \Cref{thm:wcol} will follow from  \Cref{lem:wcol_down,lem:wcol_up}, which we will prove next. 

\begin{lemma}
\label{lem:wcol_up}
Let $G$ be a graph and $d$ a positive integer. Then
\[
\herdisc(\Ss^E(G^d))<(2d\wcol_{d-1}(G)+1)\wcol_d(G).
\]
\end{lemma}

\begin{proof}
Let $L$ be a linear ordering such that
$\max_{v\in V(G)}|\WReach_d[G,L,v]|=\wcol_d(G)$. For brevity, 
in the following we write $\WReach_i[v]$ for $\WReach_i[G,L,v]$. 
For a vertex
$v\in V(G)$ we define $W_0(v)=\{v\}$ and, for $1\leq i\leq d$,
$W_i(v)=\WReach_i[v]\setminus\WReach_{i-1}[v]$. Hence
$\left(W_0(v),\dots,W_d(v)\right)$ is a partition of $\WReach_d[v]$.
In particular, $\sum_{i=0}^d|W_i(v)|=|\WReach_d[v]|\leq\wcol_d(G)$.
 We further
define, for $0\leq i\leq d$,
$\WReach_i^\ast[v]=\{u: v\in \WReach_i[u]\}$.

For every vertex $v$, we denote by $B_d[v]$ the set of all vertices
at distance at most $d$ from~$v$. 
Note that $B_d[v]=N_{G^d}(v)\cup\{v\}$. 
 For each pair $(u,v)$ of vertices at
distance at most $d$ (including pairs $(u,u)$) we fix a shortest
path $P_{u,v}$ linking $u$ and $v$, and we denote by $m_{u,v}$ the
minimum vertex of $P_{u,v}$ with respect to $L$. 
 Note that if
${\rm dist}(m_{u,v},v)=i$, then
$m_{u,v}\in W_{i}(v)\cap\WReach_{d-i}[u]$, as otherwise there would
exist a shorter path linking~$u$ and~$v$. 
If $W_{i}(v)\mathop{\cap} \WReach_{d-i}[u]$ is not empty, then $u\in B_d[v]$. Thus,
\[
B_d(v)=\bigcup_{i=0}^d \bigcup_{z\in W_{i}(v)}\WReach_{d-i}^*[z].
\]

Let $\Ss'=\{\WReach_{i}^*[z]: z\in V(G), 1\leq i\leq d\}$.
 If $u\in \WReach_{i}^*[z]$, then
$z\in\WReach_{i}[u]$, thus every vertex belongs to at most
$d\cdot\wcol_d(G)$ sets $\WReach_{i}^*[z]$. 
According to the 
Beck-Fiala Theorem, for every subset $X$ of vertices there exists a function
$\chi\colon X\to\{-1,1\}$ such that
$|\disc_\chi(\WReach_{i}^*[z]\cap X)|\leq 2d\cdot\wcol_d(G)-1$ for all
$z\in V(G)$ and $1\leq i\leq d$.  

\smallskip
As $\bigcup_{z\in W_d(v)}\WReach^*_0[z]=W_d(v)\subseteq \WReach_d[v]$, for every vertex $v\in V(G)$, we have
\begin{align*}
	\disc_\chi(N_{G^d}(v))&\leq\Bigl(\sum_{i=0}^{d-1}\sum_{z\in W_{i}(v)}\disc_\chi(\WReach_{d-i}^*[z]\cap X)\Bigr)+\disc_\chi((\WReach_d[v]\setminus\{v\})\cap X)\\
	&\leq \wcol_{d-1}(G)(2d\wcol_d(G)-1)+(\wcol_d(G)-1)\\
	&<  (2d\wcol_{d-1}(G)+1)\wcol_d(G).\qedhere
\end{align*}
\end{proof}
\begin{lemma}
\label{lem:wcol_down}
	Let $G$ be a graph and let $d$ be a positive integer. Then there exists a subgraph $H$ of $G$ and an integer $d'<2d$ with
\[
\herdisc(\Ss^E(H^{d'}))\geq \frac{1}{12d}\log_2({\rm wcol}_d(G))-\frac{1}{3}\log_2 d-\frac{11}{6}.
\]
\end{lemma}

\pagebreak
\begin{proof}
We make use of the following inequalities
relating the weak coloring numbers and shallow topological minor average degrees to $d$-admissibility\footnote{We refer the interested reader to \cite{grohe2018coloring} for a definition of $d$-admissibility. Note that the definition of  ${\rm adm}_d(G)$ and~$\wcol_d(G)$ used in our paper and the two quoted papers differs by $1$ from the original definition of \cite{Dvovrak2011}.}
 \cite[Theorem 4.1.3]{siebertz2016nowhere} and \cite{grohe2018coloring}.
\begin{align*}
{\rm wcol}_d(G)&\leq {\rm adm}_d(G)^d\\
{\rm adm}_d(G)&\leq 6d\,\lceil\widetilde{\nabla}_{d-1}(G)\rceil^3.
\end{align*}

Let $\alpha=\bigl(\frac{1}{6d}\bigr)^{1/3}\,{\rm wcol}_d(G)^{{1/3d}}\,$.
Let $M$ be a shallow topological minor of $G$ at depth $(d-1)$ witnessing
$\widetilde{\nabla}_{d-1}(G)\geq \alpha$. Each edge of $M$ appears in $G$ as 
a path of length $i+1$, for some $0\leq i\leq 2(d-1)$. Thus, for some $1\leq d'< 2d$, the graph $G$ contains the exact $(d'-1)$-subdivision $H_0$ of a graph $M_0$ with 
$|E(M')|/|V(M)|\geq \alpha/2d$. It follows that $\deg(M_0)\geq \alpha/d$.
According to \Cref{thm:discdeg}, $M_0$ has a subgraph $M$ with
$\disc(\Ss^E(M))\geq \frac{1}{4}\log_2(\pi\alpha/d)-2$. 
The $(d'-1)$-subdivision $H$ of $M$ is a subgraph of $H_0$, thus a subgraph of $G$, and $\disc(\Ss^E(M))\leq\herdisc(\Ss^E(H^{d'}))$.
Hence, the result follows, as
\begin{align*}
	\frac{1}{4}\log_2(\pi\alpha/d)-2&=\frac{1}{4}\Bigl(\log_2(\pi/d)
	-\frac{1}{3}\log_2(6d)+\frac{1}{3d}\log_2({\rm wcol}_d(G))
	\Bigr)-2\\
	&=\frac{1}{12d}\log_2({\rm wcol}_d(G))-\frac{1}{3}\log_2 d+\Bigl(
	\frac{\log_2(\pi)}{4}-\frac{\log_2 6}{12}-2\Bigr)\\
	&\geq\frac{1}{12d}\log_2({\rm wcol}_d(G))-\frac{1}{3}\log_2 d-\frac{11}{6}.
\end{align*}
\end{proof}

\begin{proof}[Proof of  \Cref{thm:wcol}]
	If $H\subseteq G$ and $d'\leq d$, then $\wcol_{d'}(H)\leq \wcol_d(G)$. Thus,  we deduce the upper bound of \Cref{thm:wcol} from \Cref{lem:wcol_up}.
	
According to \Cref{lem:wcol_down} applied to $G$ and $\lceil d/2\rceil$ (instead of $d$), there exists a subgraph $H$ of~$G$ and an integer $d'\leq d$ with
\begin{align*}
	\herdisc(\Ss^E(H^{d'}))
	&\geq \frac{1}{12 \lceil d/2\rceil}\log_2(\wcol_{\lceil d/2\rceil}(G))-\frac{1}{3}\log_2 \lceil d/2\rceil-\frac{11}{6}\\
	&\geq \frac{1}{6 (d+1)}\log_2(\wcol_{\lceil d/2\rceil}(G))-\frac{1}{3}\log_2(d+1)-\frac{3}{2}.
\end{align*}
\end{proof}

\section{A model theoretic approach to degeneracy: pointer structures}
\label{sec:ps}
A \emph{pointer structure} is a unary structure, that is, a structure with a (possibly infinite) signature~$\sigma$ that consists of unary relation and unary function symbols. 
It will be convenient to
assume that the function symbols in $\sigma$ are the elements of a monoid~$\Mon$. We fix the signature $\sigma$ and we will consider only $\sigma$-structures
satisfying the theory $T_\sigma$ that consists of the sentences
$\forall x\,f(g(x))=h(x)$, where~$f$ and~$g$ range over~$\Mon$ and $h=f\mathop{\circ} g$. Note that this is not really a restriction, but rather a way to simplify the notations.

\begin{observation}
	\label{obs:deg}
	For every $d$-degenerate graph $G$ there is a pointer structure $\mathbf M$ with $d$ unary functions such that $\Ss^E(G)=\Ss^\eta(\mathbf M)$, where $\eta(x,y)$ is the quantifier-free formula 
	\[\eta(x,y):=\neg(x=y)\wedge\bigvee_{i=1}^d\bigl((f_i(x)=y)\vee (f_i(y)=x)\bigr).\]
\end{observation}
\begin{proof}
	If $G$ is $d$-degenerate, then $G$ has an orientation $\vec G$ with maximum out-degree $d$. 
	Consider such an orientation~$\vec G$ of $G$ and an arbitrary labeling of the arcs of $\vec G$ by integers between $1$ and~$d$ such that the arcs with the same origin have pairwise different labels.
	We define $\mathbf M$ as the pointer structure with domain $V(G)$ such that for every vertex $v$ and every integer $i\in[d]$ we have $f_i(v)=w$ if $(v,w)$ is an arc of $\vec G$ with label $i$, and $f_i(v)=v$ if no arc with origin $v$ has label $i$. That $G\models E(u,v)$ is equivalent to $\mathbf M\models \eta(u,v)$ is straightforward.
\end{proof}

For a pointer structure $\strM$ (with signature $\sigma$) and a subset $A$ of the domain of $\strM$, we define the pointer structure $\strM[A]$ \emph{weakly induced} by $A$ on $\strM$ as the pointer structure with domain $A$ (and signature $\sigma$), such that for every relation symbol 
$P\in\sigma$, every function symbol $f\in \sigma$, and all distinct $u$ and $v$ in $A$, we have
\begin{align*}
	\strM[A]\models P(u)&\iff \strM\models P(u)\\
	\strM[A]\models f(u)=v&\iff \strM\models f(u)=v
\end{align*}
(It follows that if $u\in A$ and the image by $f$ of $u$ in $\strM$ is not in $A$, then $u$ is a fixed point of $f$ in~$\strM[A]$.)
Note that the Gaifman graph of $\strM[A]$ is the subgraph of the Gaifman graph of $\strM$ induced by $A$.

\subsection{Hereditary discrepancy of QF-definable sets systems}
The goal of this section is to prove the following theorem.
\begin{theorem}
	\label{thm:QF}
	For every finite set
	$\Phi=\{\phi_1(\bar x;\bar y),\dots,\phi_p(\bar x;\bar y)\}$ of quantifier free partitioned $\sigma$-formulas (with same $|\bar x|$ but possibly different $|\bar y|$) there exists a
	constant $C_\Phi$ such that for every $\sigma$-structure $\strM$
	we have
	\[
	\herdisc(\mathscr S^\Phi(\strM))\leq C_\Phi.
	\]
\end{theorem}

The proof goes by the following key steps. One of the building blocks is the Beck-Fiala Theorem. However, we cannot directly apply this result as it requires set systems to have bounded degree. We generalize the theorem to set systems that are Boolean combinations of set systems with bounded degree (Lemmas~\ref{lem:inter} and~\ref{lem:BC}).
Finally, we prove that we can decompose set systems defined by quantifier free formulas into Boolean combinations of set systems that have bounded degree (Lemma~\ref{lem:qfdeg}), which concludes the proof of \Cref{thm:QF}.

The two following lemmas extend the Beck-Fiala Theorem to handle 
set systems built from sets in a set system of bounded degree.

\begin{lemma}
	\label{lem:inter}
	Let $(U,\mathscr S_0)$  be a set system of degree at most $t$.
	Then the set system $(U,\Ss_1)$, where~$\Ss_1$ is the collection of all intersections of sets in $\Ss_0$ (with the convention that $U$ is the empty  intersection) has degree at most $2^t$.
\end{lemma}
\begin{proof}
	As each $v\in U$ belongs to at most $t$ sets  $S_1, \dots,S_t$ in~$\mathscr S_0$, it belongs to at most $2^t$ sets  in~$\Ss_1$, namely the sets of the form  $\bigcap_{i\in I}S_i$, for $I\subseteq [t]$.  Thus the set system  $(U,\Ss_1)$ has degree at most $2^t$.
\end{proof}
\begin{lemma}
	\label{lem:BC}
	Let $(U,\mathscr S_1)$ and $(U,\mathscr S)$ be two set systems over the same ground set $U$ such that $\Ss_1$ is closed under taking intersections (including $U$, the empty intersection) and such that
	every set in~$\mathscr S$ is a Boolean combination of at most  $k$ sets from $\mathscr S_1$.
	Then $\herdisc(\mathscr S)\leq 4^k\,\herdisc(\Ss_1)$.
\end{lemma}
\begin{proof}
	Let $S\in\Ss$. By assumption, $S$ is a Boolean combination of  at most $k$ sets $S_1,\dots,S_k$ in~$\Ss_1$. This Boolean  combination can be written as a disjoint union of at most $2^k$  terms of the form  $T_I=\bigcap_{i\in I}S_i\cap \bigcap_{j\in [k]\setminus I}(U\setminus S_j)$, where $I\subseteq [k]$. Let  $\mathbb I_Y\colon U\rightarrow\{0,1\}$ denote the indicator function of a  subset $Y$ of $U$.  For a subset $J$ of $[k]$, let  $S_J=\bigcap_{j\in J}S_j$.  Note that $S_J\in\Ss_1$. Then we  have
	\[
	\mathbb I_{T_I}= \mathbb I_{S_I}\cdot \prod_{j\in [k]\setminus
		I}(\mathbb I_U-\mathbb I_{S_j}) =\sum_{J\subseteq [k]\setminus
		I}(-1)^{|J|}\,\mathbb I_{S_{I\cup J}}.
	\]
	
	To see this, consider some $u\in U$. Assume first that $u$ does not appear in any of the sets~$S_j$ for $j\in [k]\setminus I$. Then the only term in the sum with $\mathbb I_{S_{I\cup J}}(u)\neq 0$ is that with $J=\emptyset$, so we obtain exactly $I_{S_{I}}$. Otherwise, the set $K = \{j\in [k]\setminus I : u\in S_j\}$ is non-empty and $u\in S_J$ if and only if $J\subseteq K$.  Thus we count $u$ exactly
	$\sum_{J\subseteq K}(-1)^{|J|}=\sum_{i=0}^{|K|}\binom{|K|}{i}(-1)^i=(1-1)^{|K|}=0$ times by the binomial theorem.
	
	Hence, for any $X\subseteq U$ we have
	$
	|\chi(T_I\cap X)| \leq \sum_{J\subseteq [k]\setminus I}|
	\chi(S_{I\cup J}\cap X)| \leq 2^k\cdot \herdisc(\Ss_1)$.
	As $\chi(S\cap X)$ is the sum of at most $2^k$ terms of the form
	$\chi(T_I\cap X)$ we deduce $\herdisc(\mathscr S)\leq 4^k\,\herdisc(\Ss_1)$.
\end{proof}

We now turn to Lemma~\ref{lem:qfdeg}, which is the main ingredient from the logical point of view.

\begin{lemma}
	\label{lem:qfdeg}
	Let $\phi(\bar x;\bar y)$ be a quantifier-free $\sigma$-formula.
	There exists a finite set
	\[\Psi=\{\rho_i(\bar x): 1\leq i\leq k\}\cup\{\psi_j(\bar x;\bar z): 1\leq j\leq\ell\}\]
	of quantifier-free $\sigma$-formulas (with same $|\bar x|$, but possibly different $|\bar z|$ for the formulas $\psi_j$)
	where
	each $\psi_j(\bar x,\bar z)$ has the form
	\[\psi_j(\bar x,\bar z):=\bigwedge_{r=1}^{|\bar z|} \bigl( f_{\alpha_{j,r}}(x_{i_{j,r}}) = z_r \bigr)\]
	and there is an
	integer~$p$ such that
	for every $\sigma$-structure~$\strM$,
	every set in $\mathscr S^\phi(\strM)$ is a Boolean combination of
	at most $p$ sets in $\Ss^\Psi(\strM)$, and $\Ss^\Psi(\strM)$ has degree at most $|\Psi|=k+\ell$.
	%
\end{lemma}

Note that we do not translate $\phi$ into an equivalent formula.
An example for the construction can be found below the proof and it may be
helpful to consider it in parallel while reading the proof.

\begin{proof}
	The formula $\phi(\bar x;\bar y)$ is equivalent
	(modulo $T_\sigma$) to a Boolean combination of
	formulas of the form
	\[
	\zeta(\bar x;\bar y):=\rho(\bar x)\wedge\eta(\bar y)\wedge
	\bigwedge_{(f,g,i,j)\in\mathcal F}f(x_i)=g(y_j),
	\]
	where $\mathcal F$ is a finite subset of
	$\Mon\times \Mon\times[|\bar x|]\times [|\bar y|]$.
	This is because we may assume that $\phi$ uses no composition of 
	functions, as $\Mon$ is a monoid, and because every atomic formula can either use a unary relation symbol (and then is in $\rho$ or $\eta$), or express equality of two terms. If both terms are of the form $f(x_i)$ (resp.\ $g(y_j)$), then this atomic formula is in $\rho$ (resp.\ in~$\eta$), otherwise it must be of the form $f(x_i)=g(y_j)$.
	Moreover, we can assume that each pair~$(g,j)$ belongs to at most one tuple $(f,g,i,j)\in\mathcal F$ as $(f(x_i)=g(y_j))\wedge (h(x_{i'})=g(y_j))$ is equivalent to $(f(x_i)=h(x_{i'})) \wedge (f(x_i)=g(y_j))$.
	Then  for every structure $\strM$ and every tuple $\bar b$ we have that~$\phi(\strM,\bar b)$ is a Boolean combination of sets
	of the form~$\zeta(\strM,\bar b)$.

	
	Now for every $\zeta(\bar x;\bar y)=\rho(\bar x)\wedge\eta(\bar y)\wedge\bigwedge_{(f,g,i,j)\in\mathcal F}f(x_i)=g(y_j)$ of the Boolean combination, we add $\rho$ and $\psi$ to $\Psi$, where we define $\bar z$ and $\psi(\bar x; \bar z)$ from $\bigwedge_{(f,g,i,j)\in\mathcal F}f(x_i)=g(y_j)$ by replacing each appearing
	term $g(y_j)$ by a new variable $z_{g,j}$.
	Note that $\bar z$ may be longer than~$\bar y$, however its length is at most the number of conjuncts in the big conjunction in~$\zeta$.
	Then $\psi$ has the required form.
	Moreover, for every tuple $\bar b$, we can define an assignment $\bar c$ of 
	$\bar z$ (let $c_{g,j} :=g(b_j)$) such that
	\[\zeta(\strM;\bar b)=\begin{cases}
		\rho(\strM)\cap\psi(\strM;\bar c)&\text{if }\strM \models\eta(\bar b)\\
		\emptyset&\text{otherwise.}
	\end{cases}\]
	
	Finally, observe that for every tuple $\bar a$, there do not exist two distinct tuples $\bar c,\bar c'$ with \linebreak \mbox{$\bar a\in \psi(\strM, \bar c)\mathbin{\cap} \psi(\strM, \bar c')$}. If $\bar a\in \psi(\strM, \bar c)\mathbin{\cap} \psi(\strM, \bar c')$, then for all $g,j$ of the conjunction,
	\linebreak \mbox{$c_{g,j}  = f(a_{i})= c_{g,j'}$}, hence $\bar c = \bar c'$.
	Hence the degree of each $\Ss^\psi$ is $1$. Obviously, the degree of each 
	$\Ss^\rho$ is also $1$ (as this set system contains
	only one set).
	As we have $k$ formulas $\rho$ and $\ell$ formulas $\psi$, we have that the degree of~$\Ss^\Psi(\mathbf{M})$
	is at most $k+\ell$.
\end{proof}

\begin{example}
	Let $\phi(x_1,x_2,x_3;y_1,y_2) = R(x_3) \wedge B(y_1) \wedge (h(x_1)=f(y_2)) \wedge  (h(x_2)=g(y_2))$.\linebreak
	We replace $f(y_2)$ by $z_1$, and $g(y_2)$ by $z_2$, and obtain
	$\rho(\bar x) = R(x_3)$ and
	$\psi(\bar x;\bar z) = (h(x_1)= z_1)$ $\wedge  (h(x_2)=z_2)$.
	
	Let $\strM$ be a $\sigma$-structure, and let $b_1,b_2$ be two elements of the universe of $\strM$, from which we derive $c_1=f(b_2)$ and $c_2=g(b_2)$. Note that for any tuple $\bar a$, if $\strM \models \phi(\bar a;\bar b)$, then $\strM \models \rho(\bar a) \wedge \psi(\bar a;\bar c)$. Note that the inverse is not true, for example if $B(b_1)$ is not satisfied. Nonetheless, for any fixed $b_1,b_2$, the set $\phi(\strM,\bar b)$ is either empty or equal to $\rho(\strM) \cap \psi(\strM,\bar c)$.
	
	For the degree bound, take two tuples $\bar c$ and $\bar c'$ such that $\psi(\strM,\bar c)\cap \psi(\strM,\bar c') \neq \emptyset$.
	Take a tuple~$\bar a$ in the intersection. We have that $h(a_1)=c_1=c'_1$ and $h(a_2)=c_2=c'_2$, hence $\bar c = \bar c'$.
	So all sets in $\psi(\strM,\bar z)$ are pairwise disjoint.
\end{example}

\begin{example}
	Let $\phi(x;y)=A(x)\wedge (B(y)\wedge (f(x)=y) \vee \neg B(y)\wedge (f(x)=f(y)))$.
	Then we can choose $\Psi=\{\rho(x),\psi(x,z)\}$, where $\rho(x)= A(x)$ and $\psi(x,z)= (f(x)=z)$, as
	for every $b$ we have
	\[
	\phi(G,b)=\begin{cases}
		\rho(G)\cap \psi(G,b)&\text{if }b\in B(G)\\
		\rho(G)\cap \psi(G,f(b))&\text{otherwise.}
	\end{cases}
	\]
\end{example}

We finally deduce a uniform bound for set systems defined on $\sigma$-structures by quantifier-free formulas, the main result of this section.

\begin{proof}[Proof of \Cref{thm:QF}]
	According to~\Cref{lem:qfdeg} (applied to each formula in $\Phi$) there exists a finite set $\Psi$ of quantifier free partitioned formulas and integers $k,t$ such that for every $\sigma$-structure~$\strM$ the set system $\Ss^{\Psi}(\strM)$ has degree at most $t$, and every set in $\Ss^\Phi$ is a Boolean combination of at most $k$ sets in $\Ss^{\Psi}$. 
	Let $\strM$ be any $\sigma$-structure and
	let~$\Ss_0$ be the the collection of all intersections of sets in $\Ss^{\Psi}(\strM)$. According to \Cref{lem:inter}, the set system~$\Ss_0$ has degree at most $2^t$. Thus, 	according to \Cref{thm:BF} we have
	$\herdisc(\Ss_0)<2^{t+1}$. Then it follows from \Cref{lem:BC} that $\herdisc(\Ss^{\Phi})<4^k\,\herdisc(\Ss_0)<2^{2k+t+1}$, so we can choose   $C_\Phi=2^{2k+t+1}$.
\end{proof}

Note that this theorem is tight in the sense that for every signature
$\sigma$ with a relation or a function of arity at least two there
exists a quantifier-free formula with unbounded discrepancy on the class of all $\sigma$-structures.
\subsection{Pointer structures and set-defined graphs}
\label{sec:setdef}
Pointer structures are intrinsically related to set-defined classes.
Recall that the \emph{age} of an infinite graph is the class of all its finite induced subgraphs. A class $\Cc$ (of finite graphs) is {\em set-defined} if 
it is included in the age of  an (infinite) graph definable in~$\mathbb N$ (seen as a model of an infinite set in the empty language)~\cite{Jiang2020}. 
In other words, $\mathscr C$ is set defined if there exists a partitioned formula $\phi(\bar x;\bar y)$ with $|\bar x|=|\bar y|=d$ such that every graph in~$\mathscr C$ is an induced subgraph of the graph with vertex set $\mathbb N^d$ in which $\bar u$ and $\bar v$ are adjacent if and only if $\mathbb N\models \phi(\bar u,\bar v)$. Set-defined classes have bounded VC-dimension (they are actually both semi-algebraic and stable). 
Set-defined classes of graphs generalize classes of degenerate graphs. 
Indeed, it is proved in \cite{Jiang2020} that 
a class is degenerate if and only if it is \emph{biclique-free} (i.e.\ excludes some biclique as a subgraph) and set-defined.

\begin{lemma}
	\label{lem:setdef}
	If a set system  is definable by a quantifier free formula in a pointer structure, then it is included in the neighborhood set system of a set-defined graph. Conversely, every neighborhood set system of a set-defined graph is included in a set system definable by a quantifier free formula in a pointer structure.
\end{lemma}
\begin{proof}
	Let $q(\bar x;\bar y)$ be a quantifier-free formula in the language of pointer structures. 
	We can assume that the formula $q$ contains no composition of functions (as we assumed that the functions in the signature form a monoid). Let $f_1,\dots,f_a$ (resp.\ $P_1,\dots,P_b$) be the functions (resp.\ the predicates) appearing in $q$. 
Let $\strM$ be a pointer structure with universe $M$. Without loss of 
generality we assume $M=\{2,\ldots, |M|+2\}$. We will define $\strM$ on the
subset $X=\{0,1\}\cup M$ of $\N$, where we will define a bipartite graph 
such that the neighborhoods of the elements of one part correspond to the sets defined by $q$ on $\strM$. 
	
	We associate to each element $u$ of $M$ the tuples
	\begin{align*}
	\Lambda_0(u)&=(0,u,f_1(u),\dots,f_a(u),0,p_1(u),\dots,p_b(u)) \text{ and }\\
	\Lambda_1(u)&=(1,u,f_1(u),\dots,f_a(u),0,p_1(u),\dots,p_b(u)),
	\end{align*}
	where $p_i(u)=1$ if $\mathbf M\models P_i(u)$, and $p_i(u)=0$ otherwise. 
	The tuples $\Lambda_0(u)$ will belong to one part of the bipartite 
	graph and the tuples $\Lambda_1(v)$ will belong to the other part. We
	use the first coordinate to distinguish the parts. 
	
	Then it is easily checked that the formula $q(\bar x;\bar y)$ is equivalent to a quantifier free formula~$Q(\bar s,\bar t)$, where $\bar s$ consists of all the tuples $\Lambda_0(x_i)$ and $\bar t$ consists of all the tuples~$\Lambda_1(y_j)$. For
	example, to check the atomic formula $f_i(x_1)=g_j(y_2)$ we need to check
	that the $(i+2)$nd coordinate of $s_1$ is equal to the 
	$(j+2)$nd coordinate of $t_2$ (the $i+2$nd coordinate of $s_1$, which is 
	the tuple $\Lambda_0(x_1)$, contains the value $f_i(x_1)$ and analogously, 
	the $j+2$nd coordinate of $t_2$, which is 
	the tuple $\Lambda_1(y_2)$, contains the value $f_j(y_2)$). Similarly, 
	$\strM\models P_i(x_1)$ if and only if the coordinate $a+3+i$ of 
	$s_1$ contains $1$. We use the $0$ in coordinate $a+3$ to test this, by 
	checking if the coordinate $a+3$ of $s_1$ is equal to the coordinate $a+3+i$ 
	of  $s_1$. 
	
	The set system defined by $q$ on $\mathbf M$ is the neighborhood set system of an induced subgraph of the bipartite graph defined on $X$ (as 
	a subset of~$\N$) by $\bar s$ adjacent to $\bar t$ if $X\models Q(\bar s;\bar t)$.
	
	\pagebreak
	Conversely, assume that $\phi(\bar x;\bar y)$ is a first-order formula (in the empty language). We may assume $|\bar x|=|\bar y|=d>1$, as we can repeat a coordinate and pad to get $|\bar x|=|\bar y|>1$ if needed. We may also assume that $\phi$ is quantifier free (as the theory of infinite sets has quantifier elimination, see~\cite{hodges1997shorter}). Let $G$ be a 
	finite induced subgraph of the graph defined by $\phi$ on $\N$. Let $A$
	be the set of vertices of $G$, which is a finite subset of $\N^d$. 
	  Let $\mathbf M$ be the pointer structure with domain $A\cup X$, where $X$ is the set of all elements of $\N$ appearing in a tuple
	of $A$, and where $f_i\colon A\cup X\rightarrow X$ is defined by: 
	$f_i(\bar v)=v_i$ if $\bar v\in A$ and 
	$f_i(u)=u$ if $u\in X$. Note that an element of $A\cup X$ is a fixed-point of~$f_i$ if and only if it belongs to $X$. 
	It is easily checked that $\Ss^E(G)$ is included in~$\Ss^q(\mathbf M)_{|A}$ for some quantifier free formula $q(x,y)$. 
\end{proof}

\begin{corollary}
	\label{cor:RPS}
	Let $(U,\Ss)$ be a set system definable in an infinite set $X$.
	Then there exists a constant $C$ such that for every finite subset $A$ of $U$ we have
	$\disc(\Ss_{|A})\leq C$.
\end{corollary}

\begin{theorem}
	\label{thm:setdef}
	For every set-defined class $\Cc$ there exists a constant $C$ with 
\[
	\herdisc(\Ss^E(G))\leq C,
\]
for every $G\in\Cc$.
\end{theorem}
\begin{proof}
The statement is a direct consequence of \Cref{lem:setdef,cor:RPS}.
\end{proof}

We now use \Cref{ex:Sp} to give a lower bound on the discrepancy of neighborhood set systems of set-definable graphs in terms of the length of the parameter tuples used in the formula defining the set system.
\begin{example}
	\label{ex:setdef}
	For every non-negative integer $d$ there exists a set system $(U,\Ss)$ definable in an infinite set $X$ by a formula $\phi(\overline{x};\overline{y})$ with $|\overline{x}|=|\overline{y}|=d+1$, such that $\disc(\Ss_{|A})= \Omega(\sqrt{d})$ for some subset $A$ of size $2^{d+1}$.
\end{example}
\begin{proof}
	For a non-negative integer $p$, let $S_p$ be Sylvester's graph, 
	introduced in \Cref{ex:Sp}.
	The above statement will follow from the lower bound on $\disc(\Ss^E(S_p))$ given there and the
	property that  $S_p$ is an induced subgraph of the graph defined on $\mathbb N$ by some formula $\phi_p(\overline{x},\overline{y})$, where
	$|\overline{x}|=|\overline{y}|=p+1$, which we prove now.

	We define inductively the parts $U_p$ and $V_p$ of $S_p$, as well as the formula $\phi_p$.
	We let \mbox{$U_0=\{1\}$}, $V_0=\{2\}$, and $\phi_0(x,y)=\neg (x=y)$.
	Assume $S_k$ is the subgraph of the graph defined by $\phi_k$ induced by $U_k$ and $V_k$.
	We define
\begin{align*}
	 \phi_{k+1}(x_1,\dots,x_{k+1},y_1,\dots,y_{k+1})&=\phi_k(x_1,\dots,x_k,y_1,\dots,y_k)\leftrightarrow ((x_{k+1}=x_1)\vee(y_{k+1}=x_1)),\\ 
	U_{k+1}&=\{(x_1,\dots,x_{k+1})\in\{1,2\}^{k+1}: x_1=1\},\\
	\text{and } V_{k+1}&=\{(y_1,\dots,y_{k+1})\in\{1,2\}^{k+1}: y_1=2\}.
\end{align*}	
\end{proof}
\pagebreak

\section{A model theoretic approach to classes with bounded expansion}
\label{sec:pointer-theory}

In this section we translate the
quantifier elimination scheme introduced in \cite{DKT2} for classes
with bounded expansion into a model theoretic language, which allows us
to extend the result obtained in \Cref{sec:wcol} to the general
setting of set systems definable in graphs of a bounded expansion
class. 
Important model theoretic properties of classes of graphs with bounded expansion (like the existence of a quantifier elimination scheme) are based on the existence of so-called ``transitive fraternal augmentations'' with bounded in-degree. Here, 
we work with transitive fraternal augmentations with \emph{bounded out-degree} instead, as arcs 
$(u,v)$ will represent functions and bounded out-degree
corresponds to a bounded number of functions in a structure. 
This can be formalized as follows \mbox{(see \cite[Section 7.4]{Sparsity}} and also \cite{POMNI}): A class $\Cc$ of finite graphs has bounded expansion if and only if  there exists a function $F\colon\mathbb N\rightarrow\mathbb N$ such that for every $G\in\Cc$ there exists a directed supergraph 
$\DG$ and a \emph{fraternity function} $r\colon E(\DG)\rightarrow \mathbb N$ satisfying:
\begin{enumerate}[(P1)]
	\item\label{it:Gstart}{\bf founding:}  an arc $e=(u,v)$ of $\DG$ is such that $r(e)=1$ if and only if it links two vertices that are adjacent in $G$;
	\item\label{it:Gfrat} {\bf fraternity:}  for every two arcs $e_1=(u,v)$ and $e_2=(u,w)$ of $\DG$ there is an arc $e_3=(v,w)$ or $(w,v)$ in $\DG$ with $r(e_3)\leq r(e_1)+r(e_2)$; 
	\item\label{it:Gbound} {\bf boundedness:} every vertex of $\DG$ has at most $F(i)$ outgoing arcs $e$ with $r(e)\leq i$.
\end{enumerate}	

Note that the fraternal augmentation algorithm presented in \cite{Sparsity} (see also \cite{POMNII}) can be iteratively used to compute a supergraph $\DG$ and a fraternity function $r$ of a finite graph~$G$  in time~$O(n^4)$ satisfying the properties \labelcref{it:Gstart,it:Gfrat}. Moreover, for every class $\Cc$ of
finite graphs with bounded expansion there exists a function $F_\Cc$ such that for graphs in $\Cc$ the property \labelcref{it:Gbound} will be satisfied by the computed~$\DG$ and $r$ (with $F=F_\Cc$). Note that we do not use 
the finiteness of $G$ in the definition of fraternity functions and we will 
use all defined notation also for infinite graphs.

We now introduce theories ${\rm Th}^{\rm TF}_{\sigma,\rho}$, which  are intrinsically related to the notion of bounded expansion. 
A \emph{ranked signature} is a pair $(\sigma,\rho)$, where $\sigma$ is a signature and $\rho\colon \sigma\rightarrow\mathbb N$ is a \emph{(signature) ranking}.
A \emph{ranked} pointer structure is a pointer structure with a ranked signature.
For a positive integer $i$, we denote by $\sigma_i$ the subset of symbols in $\sigma$ with rank at most $i$, and we define
the \emph{rank-$i$ shadow} of a $\sigma$-structure $\strM$ as the  $\sigma_i$-reduct of $\strM$, that is the $\sigma_i$-structure obtained from $\strM$ by ``forgetting'' the functions and relations whose symbols are not in~$\sigma_i$.
Note that the structures $\strM$ considered in this section may be infinite.

We consider signatures
\begin{align*}
\varsigma&=\mathcal F_\sigma\cup \{M_\alpha\colon \alpha\in\mathbb N\}\\
\sigma&=\varsigma\cup\{P_\phi\colon \phi\in{\rm FO}_1(\varsigma)\},
\end{align*}
where $\Mon$ is a monoid of function symbols, the symbols $M_\alpha$ and $P_\phi$ are unary predicates, and ${\rm FO}_1(\varsigma)$ is the class of all (equivalence classes of) first-order formulas on $\varsigma$-structures with a single free variable\footnote{By ``equivalence classes'' we mean that we keep only one representative of each equivalence class for the relation ``is logically equivalent to''.}. Moreover, we require that the  rank function $\rho$ on $\sigma$ is such that every $\sigma_i$ is finite. Note that this implies that $\sigma$ is countable.
 
 \pagebreak
We define
the theory ${\rm Th}^{\rm TF}_{\sigma,\rho}$ of \emph{fraternal} $\sigma$-structures by the following axioms, which mirror the properties of fraternity functions.
\begin{enumerate}[({A}1)]
	\item\label{it:frat} {\bf fraternity:} For all functions $f,g\in\Mon$, and for every $x$ there exists a function $h$ (possibly depending on $x$) with $h(f(x))=g(x)$ or $h(g(x))=f(x)$ and $\rho(h)\leq \rho(f)+\rho(g)$.
	
	As each $\sigma_i$ is finite, this axiom can be expressed by the sentences $\theta_{i,j}^{\rm frat}$ (for $i,j\in\mathbb N$), where 
	\[\theta_{i,j}^{\rm frat}\coloneqq
	\bigwedge_{f\in\sigma_i}\bigwedge_{g\in\sigma_j}\Bigl(\forall x\ \bigvee_{h\in\sigma_{i+j}}(h(x)=f(g(x))\vee (h(x)=g(f(x))\Bigr).\]
	
	\item\label{it:col} {\bf transitive coloring:} The unary predicates $M_\alpha$ with $\rho(M_\alpha)=i$ define a partition of the domain, and for every sequence of $k \leq i$ functions $f_1,\dots,f_k$ where $\rho(f_j)\leq i$ for all $j\leq k$ we have
		$M_\alpha(x)\wedge(f_1\circ\dots\circ f_k(x)\neq x)\rightarrow\neg M_\alpha(f_1\circ\dots\circ f_k(x))$.
	\item\label{it:typ}  {\bf typing:} For every formula $\phi\in{\rm FO}_1(\varsigma)$  the predicate $P_\phi$ satisfies  $P_\phi(x)\leftrightarrow \phi(x)$.
\end{enumerate}

Recall that a \emph{$p$-centered coloring} of a graph $G$ is a coloring of the vertices such that every connected subgraph with at most $p$ colors on the vertices contains some vertex with a unique color. 

\begin{lemma}
	\label{lem:ltdM}
Let $\strM$ be a fraternal $\sigma$-structure and let $p,q$ be positive integers. Then the unary predicates~$M_\alpha$ with 
$\rho(M_\alpha)=q\cdot (2^{p-1}+2)$ define a $(p+1)$-centered coloring of the Gaifman graph of the rank-$q$ shadow of $\strM$.
\end{lemma}
\begin{proof}
Denote by $\vec{G}$ the (possibly infinite) directed graph whose vertex set is the domain of $\strM$, in which there is an arc from a vertex $u$ to a (distinct) vertex $v$
if there exists $f\in \mathcal F_\sigma$ with $f(u)=v$ and $\rho(f)\leq q$, and let $G$ be its underlying undirected graph (which is the Gaifman graph of the rank-$q$ shadow of $\strM$). Note that for every pair of distinct vertices $u$ and $v$, there is at most one arc from $u$ to $v$ and at most one arc from $v$ to $u$ since we consider simple graphs only. Let $\vec G^+$ be the directed graph obtained in a similar way but without the constraint $\rho(f)\leq q$.  For each arc $e=(u,v)$ of $\vec G^+$ define $r(e)=\lceil \rho(f)/q\rceil$, where $f\in\mathcal F_\sigma$ is such that $f(u)=v$ and $\rho(f)$ is minimal for this property. This implies that the axiom~(P1) is  satisfied by $\vec{G}^+$ and $r$. 
Then, for every two arcs $e_1=(u,v)$ and $e_2=(u,w)$ of $\vec{G}^+$, we deduce from axiom~(A1) that  at least one of $(v,w)$ and $(w,v)$ is an arc of $\vec{G}^+$ with $r$-value at most $r(e_1)+r(e_2)$. It follows that the axiom (P2) is satisfied by $\vec{G}^+$ and~$r$. 
The axiom (P3) is also satisfied by $\vec G^+$ and~$r$, with $F(i)=|\{f\in\mathcal F_\sigma\colon \rho(f)\leq i\}|$. Then, it follows from \cite[Lemma 7.8]{Sparsity} that   every vertex coloring $c$ with $c(u)\neq c(v)$ whenever there is a directed path of length at most $p$ using only arcs $e$ with $r(e)\leq 2^{p-1}+2$ linking distinct vertices $u$ and $v$ is a $(p+1)$-centered coloring of $G$. 
Note that the proof of  \cite[Lemma 7.8]{Sparsity} does not involve any argument based on the finiteness of the considered graphs, and thus remains valid in an infinite graph setting. The considered arcs correspond to functions with
rank $\rho$ at most $q\cdot(2^{p-1}+2)$ and we conclude that the predicates
$M_\alpha$ with $\rho(M_\alpha)\leq q\cdot(2^{p-1}+2)$ define such a coloring, 
as desired.
\end{proof}

The connection between ranked pointer structures models of ${\rm Th}^{\rm TF}_{\sigma,\rho}$ and classes with bounded expansion follows.
\begin{lemma}
	\label{lem:RPStoBE}
	A class $\mathscr C$ (of finite graphs) has bounded expansion if and only if there exists a ranked signature $(\sigma,\rho)$ such that for every graph $G\in\mathscr C$ there exists a model $\strM$ of~${\rm Th}^{\rm TF}_{\sigma,\rho}$ with the property that $G$ is the Gaifman graph of the rank-$1$ shadow of $\strM$.
\end{lemma}
\begin{proof}
	Assume $\Cc$ has bounded expansion, $G\in\Cc$ and $\vec{G^+}$ is a directed supergraph of $G$ with properties (P1) to (P3).
	At each vertex $u$ we number arbitrarily the outgoing arcs $e=(u,v)$ with $r(e)=i$ by distinct values $\lambda(e)\in [F(i)]$. We define the mappings $\zeta_{i,j}\colon V(G)\rightarrow V(G)$ by 
	$\zeta_{i,j}(u)=v$ if $e=(u,v)\in E(\DG)$, $r(e)=i$, and $\lambda(e)=j$; $\zeta_{i,j}(u)=u$, otherwise. As the maximum out-degree of arcs with $r$-value $i$ is $F(i)$, there is a vertex coloring $\gamma_i$ with at most $2iF(i)^i+1$ colors such that no directed path of length at most $i$ with arcs $e$ with $r(e)\leq i$ links two vertices with the same $\gamma_i$-value. To see this, observe that the $i$th 
	power of a directed graph with out-degree~$F$ has out-degree at most $iF^i$, such that the underlying undirected graph is $2iF^i$ degenerate and the greedy coloring
uses at most $2iF^i+1$ colors. 
	Let $\sigma_0=\{f_{i,j}: i\in\mathbb N, j\in [F(i)]\}\cup \{\M_{i,j}: i\in\mathbb N, j\in [(2F(i)^i+1]\}$ and $\sigma=\sigma_0\cup\{M_\psi: \psi\in\FO_1[\sigma_0]\}$. We define the rank $\rho$ on $\sigma$ by
	$\rho(f_{i,j})=\rho(M_{i,j})=i$ and $\rho(P_\psi)=1$.
	Let $\strM$ be the $\sigma$-structure with domain $V(G)$ where we interpret $f_{i,j}(x)$ by $\zeta_{i,j}(x)$,  $M_{i,j}(x)$ by $(\gamma_i(x)=j)$, and  $P_\psi(x)$ by $\psi(x)$. By construction, $\strM$ is a model of ${\rm Th}^{\rm TF}_{\sigma,\rho}$, and $G$ is the Gaifman graph of the rank-$1$ shadow of $\strM$.
	
	The converse follows from \cref{lem:ltdM} (for $q=1$) and the property that classes with bounded $p$-centered colorings have bounded expansion.
\end{proof}

%

\begin{lemma}
	\label{lem:red}
	Let $p$ be a positive integer, let $\sigma$ be the signature of a pointer structure,  and let  $\psi(\bar x)$ be a formula of the form 
	\[
	\psi(\bar x)\coloneqq \exists z_1\ldots \exists z_p \psi'(\bar x,\bar z),
	\]
where~$\psi'$ is a conjunction of atomic formulas of the form 
$z_i=f(z_j)$ or $z_i=x_j$ and a Boolean combination of 
formulas of the form $z_i=z_j$ or $\alpha(z_i)$.

Then there exists a formula $\zeta(\bar x)$, which is 
a Boolean combination of quantifier-free formulas~$q_i(\bar x)$ and single variable formulas $\eta_i(x_j)$, such that $\psi$ and $\zeta$ are  equivalent on $\sigma$-structures whose Gaifman graphs have tree-depth at most $p$. 
\end{lemma}
\begin{proof}
	By removing all the symbols that are not used in $\psi$, we can restrict to the case where the signature is finite.
	It follows from \cite[Section 6.8]{Sparsity} that for every two integers $r$ and $p$ there is an integer $C(r,p)$ such that every colored graph\footnote{By a colored graph, we mean a graph whose vertices and edges are colored using a finite set of colors.} with treedepth at most $p$ is $r$-equivalent to one of its induced subgraphs of order at most $C(r,p)$.
	The lemma thus follows from the finite case, for which we refer the reader (for instance) to \cite{DKT2}.
\end{proof}

We continue to prove that ${\rm Th}^{\rm TF}_{\sigma,\rho}$ has quantifier 
elimination. 
Following the standard quantifier elimination proof scheme (see e.g.~\cite[Lemma 2.3.1]{hodges1997shorter}), we can reduce to the case of eliminating a single
existential quantifier.

\begin{lemma}
	\label{lem:QE}
Let $\phi(\bar x):=\exists y\, \phi'(\bar x,y)$ be a formula, where $\phi'$ is  a quantifier-free $\sigma$-formula. 
Then~$\phi(\bar x)$ is equivalent to a quantifier-free 
formula $\psi(\bar x)$ on all models of ${\rm Th}^{\rm TF}_{\sigma,\rho}$. 
\end{lemma}
\begin{proof}
Let $\sigma_1\subseteq \sigma$ be the (finite) set of symbols used in $\phi$. 
We consider an arbitrary (possibly infinite) model $\strM$ of ${\rm Th}^{\rm TF}_{\sigma,\rho}$
and construct a formula $\psi$ such that $\strM\models \phi(\bar x)\leftrightarrow\psi(\bar x)$. Our construction of $\psi$ will be 
independent of the choice of the model $\strM$. Let $p$ be the minimum 
integer such that $\phi'$ is equivalent to a formula of the form 
$\exists z_1\ldots \exists z_p\, \psi'(\bar x, y,\bar z)$,
where~$\psi'$ is a conjunction of atomic formulas of the form 
$z_i=f(z_j)$ or $z_i=x_j$ or $z_i=y$ and a Boolean combination of 
formulas of the form $z_i=z_j$ or $\alpha(z_i)$.

Let $\strM_1$ be the
$\sigma_1$-reduct of $\strM$ and 
let $q=\max\{\rho(f)\colon f\in\sigma_1\}$, and let $t=q(2^{p-1}+2)$. 
According to \cref{lem:ltdM}, the predicates $M_\alpha$ with $\rho(\alpha)=t$ define a partition of the domain of $\strM_1$, such that any $p$ classes weakly induce a pointer structure whose Gaifman graph has tree-depth at most $p$.
Let $T=\{\alpha\colon \rho(M_\alpha)=t\}$. 
We rewrite the formula $\exists y\,\bigl(\exists z_1\ldots \exists z_p\, \psi'(\bar x, y,\bar z)\bigr)$ as 
\[
\bigvee_{t_1\in T}\dots\bigvee_{t_p\in T}\exists y\exists z_1\dots\exists z_p\ \biggl(M_{t_0}(y)\wedge\Bigl(\bigwedge_{i=1}^p M_{t_i}(z_i)\Bigr)\wedge\psi'(\bar x, y,\bar z)\biggr).
\]

According to \cref{lem:red} there exist formulas $\zeta_{t_1,\dots,t_p}(\bar x)$, which are Boolean combinations of quantifier-free formulas and single variable formulas, such that the above formula is equivalent to 
\[
\bigvee_{t_1\in T}\dots\bigvee_{t_p\in T} \zeta_{t_1,\dots,t_p}(\bar x).
\]

By replacing in each $\zeta_{t_1,\dots,t_p}$ each single free variable formula $\eta(x)$ by the corresponding predicate $P_\eta$, we get a quantifier-free formula $\widehat\zeta_{t_1,\dots,t_p}(\bar x)$ that is equivalent to $\zeta_{t_1,\dots,t_p}(\bar x)$ on every model of ${\rm Th}^{\rm TF}_{\sigma,\rho}$. Thus, we have
${\rm Th}^{\rm TF}_{\sigma,\rho}\models \phi(\bar x)\leftrightarrow\psi(\bar x)$, where $\psi(\bar x)$ is the quantifier-free formula $\bigvee_{t_1\in T}\dots\bigvee_{t_p\in T} \widehat\zeta_{t_1,\dots,t_p}(\bar x)$.
\end{proof}

The main result of this section now follows by an easy induction from \cref{lem:QE}.

\begin{theorem}
\label{thm:QE}
	For every ranked signature $(\sigma,\rho)$ and every formula~$\phi(\bar x)$ there exists a quantifier free formula $q(\bar x)$ such that $\phi$ and $q$ are equivalent on every model of ${\rm Th}^{\rm TF}_{\sigma,\rho}$.\\
	In other words, the theory ${\rm Th}^{\rm TF}_{\sigma,\rho}$ has quantifier elimination.
\end{theorem}

We now can derive a model theoretical characterization of bounded expansion classes.

\begin{theorem}
\label{thm:BE_hdisc}
Let $\mathscr C$ be a monotone class of (finite) graphs.
Then the following are equivalent:
\begin{enumerate}[(i)]	
	\item the class $\mathscr C$ has bounded expansion;
	\item  the hereditary discrepancy of every set system $\mathscr S^\phi(G)$ definable on a monadic expansion of $\mathscr C$ is bounded;
	\item for each positive integer $k$, the hereditary discrepancy of $\mathscr S^E(G^k)$ for $G\in\mathscr C$ is bounded.
\end{enumerate}
\end{theorem}
\begin{proof}
	(i)$\Rightarrow$(ii) as a consequence of \Cref{lem:RPStoBE,cor:RPS},
	(ii)$\Rightarrow$(iii) as $\mathscr S^E(G^k)$ is definable on $\mathscr C$, and 
	(iii)$\Rightarrow$(i) as a consequence of \Cref{cor:BEpower}.
\end{proof}


\section{Discrepancy in nowhere dense classes}
\label{sec:ND}

Let~$\phi(\bar x,\bar y)$ be a first-order formula and let $\Cc$ be a nowhere dense class. We now give two bounds on the discrepancy of $\Ss^\phi(G)$ (for $G\in\Cc$), in terms of the lengths of the tuples $\bar x$ and $\bar y$, respectively.


\pagebreak
\begin{lemma}\label{lem:disc-nd}
Let $\Cc$ be a nowhere dense class and let $\phi(\bar x,\bar y)$ be a 
first-order formula. Then for every $\varepsilon>0$ there exists a constant $C$ with
\[
\disc(\Ss^\phi(G)_{|A})\leq C\,n^{\frac12-\frac1{2|\bar y|}+\varepsilon}
\]
for every $G\in\Cc$ and $A\subseteq V(G)^{|\bar x|}$, where $n=\bigl|\bigcup\Ss^\phi(G)_{|A}\bigr|$.
\end{lemma}
\begin{proof}
Let $m\in\mathbb N$ and let $M$ be a subset of $V(G)^{|\bar x|}$ with $|M|=m$. Then $M\subseteq U^{|\bar x|}$, where $U\subseteq V(G)$ is the set of all the elements appearing in some tuple in $M$. Obviously $|U|\leq |\bar x|\cdot|M|$. 
According to \Cref{thm:types}, for every $\varepsilon>0$ there exists a constant $c$ such that for every $G\in\Cc$ and every $A\subseteq V(G)$ we have
$\bigl|\Ss^\phi(G)_{|A^{|\bar x|}}\bigr|\leq c\,|A|^{|\bar y|+\varepsilon}$.
Thus, we have
$\bigl|\Ss^\phi(G)_{|A}\bigr|\leq c\,(m\,|\bar x|)^{|\bar y|+\varepsilon}$.
It follows that the primal shatter function $\pi(m)$ of $\Ss^\phi(G)$ 
(and, more generally, of $\Ss^\phi(G)_{|A}$ for $A\subseteq V(G)^{|\bar x|}$) 
is bounded by $c'm^{|\bar y|+\varepsilon}$, where $c'=c\,|\bar x|^{|\bar y|}$, 
and the result follows from \Cref{thm:matousek}.
\end{proof}
\begin{lemma}\label{lem:disc-nd2}
	Let $\Cc$ be a nowhere dense class and let $\phi(\bar x,\bar y)$ be a 
	first-order formula. Then for every $\varepsilon>0$ there exists a constant $C$ with
	\[
	\disc(\Ss^\phi(G)_{|A})\leq C\,n^{\frac12-\frac1{2|\bar x|}+\varepsilon}
	\]
	for every $G\in\Cc$ and $A\subseteq V(G)^{|\bar x|}$, where $n=\bigl|\bigcup\Ss^\phi(G)_{|A}\bigr|$.
\end{lemma}
\begin{proof}
	Let $\phi^\ast(\bar y;\bar x)=\phi(\bar x;\bar y)$.
As in the proof of \cref{lem:disc-nd}, the shatter function of $\Ss^{\phi^\ast}(G)$ is bounded by $c\,m^{|\bar x|+\varepsilon}$ for some constant $c$.
As $\Ss^{\phi^\ast}(G)$ is the dual of the set system of $\Ss^\phi(G)$, the
	dual shatter function $\pi^\ast(m)$ of $\Ss^\phi(G)$
	is bounded by $c\,m^{|\bar x|+\varepsilon}$ as well, 
	and the result follows from \Cref{thm:MWW}.
\end{proof}

\begin{theorem}
\label{thm:ND}
	For a monotone class of graphs $\Cc$ the following are equivalent:
	\begin{enumerate}[(1)]
		\item\label{it:ND1} $\Cc$ is nowhere dense;
				\item\label{it:ND2} for every monadic lift $\Cc^+$ of $\Cc$, for every
		formula $\phi(\bar x;\bar y)$ in the language of $\Cc^+$, for every $\alpha>\frac12-\frac{1}{2\min(|\bar x|,|\bar y|)}$, and for $G\in\Cc$,   we have
		
		\[\textstyle\herdisc(\Ss^\phi(G^+))\in\Oof\Bigl(\bigl|\bigcup \Ss^\phi(G^+)\bigr|^{\alpha}\Bigr);\]
			\item\label{it:ND2s} for every  positive integer $k$, every  $\varepsilon>0$, and for $G\in\Cc$,  we have
\[\disc(\Ss^E(G^k))\in\Oof(|G|^{\varepsilon});\]
		\item\label{it:ND3} there exists a positive integer $r$ such that for every partitioned formula $\phi(\bar x;\bar y)$ with $|\bar x|=|\bar y|=r$ and for $G\in\Cc$ we have
		$\disc(\Ss^{\phi}(G))\in o\Bigl(\bigl|\bigcup \Ss^{\phi}(G)\bigr|^{1/2}\Bigr)$.
		
	\end{enumerate}
\end{theorem}
\pagebreak
\begin{proof}
	The proof follows from the following implications:
	\begin{itemize}
		\item[\eqref{it:ND1}$\Rightarrow$\eqref{it:ND2}]  This immediately follows from \Cref{lem:disc-nd,lem:disc-nd2}.
		\item[\eqref{it:ND2}$\Rightarrow$\eqref{it:ND2s}] This immediately follows from the special case of  the set system defined by a formula~$\phi(x,y)$ expressing that the distance between $x$ and $y$ is at most $k$.  
		\item[\eqref{it:ND2s}$\Rightarrow$\eqref{it:ND1}] (by contradiction): 
		Assume the class $\Cc$ is not nowhere dense. As $\Cc$ is monotone there exists a positive integer $k$ such that the class $\Cc$ contains the $(k-1)$-subdivision of all graphs. In particular, $\Cc$ contains the $(k-1)$-subdivision of all the graphs $S_p$ (defined in \Cref{ex:Sp}).
		Let  $G_p$ be the $(k-1)$-subdivision of $S_p$ and let $A_p\subseteq V(G_p)$ be the set  of the principal vertices of $G_p$ (i.e.\ the vertices of $S_p$). Then 
		$\disc(\Ss^E(G^k))=\Omega(|A_p|^{1/2})=\Omega(|G|^{1/4})$. Hence $\neg$\eqref{it:ND1}$\ \Rightarrow\neg$\eqref{it:ND2s}.
		\item[\eqref{it:ND2}$\Rightarrow$\eqref{it:ND3}] This is trivial as we can always choose $\alpha<1/2$ in \eqref{it:ND2}.
		\item[\eqref{it:ND3}$\Rightarrow$\eqref{it:ND1}] (by contradiction):
		Assume the class $\Cc$ is not nowhere dense. As above, $\Cc$ contains the $(k-1)$-subdivision of all the graphs $S_p$ (defined in \Cref{ex:Sp}), for some positive integer $k$.
		
		For $p>1$, consider the graph $G_p$ obtained from $S_p$ by adding a new vertex $s$ adjacent to all the vertices in one of the parts of $S_p$, and a vertex $r$ adjacent only to $s$.
		Let $\mu_A(x)$ be the formula expressing that $x$ has degree greater than one and is adjacent to a vertex adjacent to a vertex with degree one, and let $\mu_B(x)$ be the formula expressing that no neighbor of $x$ is adjacent to a vertex with degree one.
		
		For $d\in\mathbb N$ we define the formula
		$\zeta_d(\bar x;\bar y)$ with $|\bar x|=|\bar y|=d$ inductively, as follows (where $\bar x'=(x_1,\dots,x_{d-1})$ and $\bar y'=(y_1,\dots,y_{d-1})$):
		$\zeta_1(x,y):=\mu_A(x)\wedge \mu_B(y)\wedge E(x,y)$ and, for $d>1$ we define 
		$\zeta_d(\bar x,\bar y):=
		\Bigl(\zeta_{d-1}(\bar x',\bar y')\leftrightarrow E(x_d,y_d)\Bigr)\ \wedge\ \bigwedge_{i=1}^d (\mu_A(x_i)\wedge\mu_B(y_i))$.
		
		It is easily checked that for $p>1$ the graph defined on $V(G_p)^d$ where $\bar a$ is adjacent to $\bar b$ if $G_p\models\zeta_d(\bar a,\bar b)$ is isomorphic to the union of $G_{dp}$ and isolated vertices. In particular,  we have $\disc(\Ss^{\zeta_d}(G_p))\in \Omega(\sqrt{|V(G_p)|^d})$.		 
		
		For $p>2$, let $H_p$ be the $(k-1)$-subdivision of $G_p$, and let $\hat\zeta_d(\bar x;\bar y)$ be the formula obtained from $\zeta_d$ by replacing $E(x,y)$ by a formula expressing that ${\rm dist}(x,y)=d$, and by conditioning all the vertices to have degree at least $3$. Then $\Ss^{\hat\zeta_d}(H_p)=\Ss^{\zeta_d}(G_p)$. Thus, $\disc(\Ss^{\hat\zeta_d}(H_p))\in \Omega\Bigl(\bigl|\bigcup \Ss^{\phi}(H_p)\bigr|^{1/2}\Bigr)$. As $\Cc$ contains all the graphs $H_p$, we deduce $\neg$\eqref{it:ND1}$\ \Rightarrow\neg$\eqref{it:ND3}.\qedhere
	\end{itemize}
\end{proof}
\begin{conjecture}
	\label{conj:ND}
A monotone class $\Cc$ is nowhere dense if and only if every for every partitioned formula $\phi(\bar x;\bar y)$, for every $\varepsilon>0$, and for $G\in\Cc$ we have
$\herdisc(\Ss^\phi(G))\in \Oof\Bigl(\bigl|\bigcup \Ss^\phi(G)\bigr|^{\varepsilon}\Bigr)$.
\end{conjecture}

\section{Applications}
\label{sec:applications}
We now consider the problems introduced in \Cref{sec:intro} in the light of the results obtained in the previous sections.
\subsection{Simultaneous neighborhood discrepancy for red and blue edges}
\label{sec:app1}

Let $G$ be a graph. To a $2$-coloring $\gamma\colon E(G)\rightarrow\{1,2\}$ we associate the set system $\Ss^\gamma(G)$ whose members are the $1$-neighborhoods and the $2$-neighborhoods of the vertices in $G$, where by \emph{$i$-neighborhood} of a vertex $v$  (with $i\in\{1,2\}$) we mean the set of all the vertices $u$ adjacent to $v$ by an edge with $\gamma$-color $i$.
For every subgraph $H\subseteq G$ we can consider the coloring $\gamma$ with $\gamma(e)=1$ if $e\in E(H)$. 
Then $\Ss^E(H)\subseteq\Ss^\gamma(G)$.
It follows that 
\[
\max_{H\subseteq G}\disc(\Ss^E(H))\leq \max_\gamma\disc(\Ss^\gamma(G)).
\]

On the other hand, let $G^\gamma$ be the  bipartite graph with parts $V$ and $W=V\times\{1,2\}$, where $u$ is adjacent to $(v,i)$ if $uv\in E(G)$ and $uv$ is colored $i$. It is easily checked that $\deg(G^\gamma)\leq \deg(G)$.
Moreover, it is easily checked that
$\Ss^\gamma(G)$ is included in $\Ss^E(G^\gamma)$. Thus, any upper bound on $\disc(\Ss^E(G^\gamma))$ based on its degeneracy provides an upper bound on $\disc(\Ss^\gamma(G))$ in terms of~$\deg(G)$.

From this and \Cref{thm:discdeg}, we deduce the following (partial) answer to \Cref{pb:discdeg}.

\begin{corollary}
	\label{cor:edgecol}
	For every graph $G$ we have
\begin{equation}
	\frac{\log_2(\pi\deg(G))}{4}-2\leq \max_\gamma \disc(\Ss^\gamma(G))<3\deg(G),
\end{equation}
where the maximum runs over $2$-colorings $\gamma\colon E(G)\rightarrow\{1,2\}$ and $\Ss^\gamma(G)$ is the set system, whose members are the $1$-neighborhoods and the $2$-neighborhoods of the vertices in~$G$.
\end{corollary}

\pbone*
\begin{solution}
	According to \cref{cor:edgecol}, we deduce 
	from $\max_\gamma \disc(\Ss^\gamma(G))\leq 1$ that the degeneracy of~$G$ is at most $2^{12}/\pi<1304$.
\end{solution}
\subsection{Epsilon-nets and epsilon-approximations.}
\label{sec:app2}
An important application of the notion of discrepancy is that of an
$\varepsilon$-net. A subset $N$ of the ground set $U$ of a set system
$(U,\Ss)$ is an \emph{$\varepsilon$-net} for $\Ss$ if $N$ intersects
all the sets $S\in \Ss$ that have at least $\varepsilon\,|U|$ elements.
A related concept is that of an $\varepsilon$-approximation. A subset
$A\subseteq U$ is an \emph{$\varepsilon$-approximation} for~$\Ss$ if
$\left|\frac{|A\cap S|}{|A|}-\frac{|S|}{|U|}\right|\leq\varepsilon$
for every $S\in\Ss$. The notions of 
$\varepsilon$-nets and $\varepsilon$-approximations
play a key role in the approximation of general sets from the systems 
by smaller subsets. 
Upper bounds for the sizes of $\varepsilon$-nets and
$\varepsilon$-approximations of a set system
can be derived from bounds on its hereditary discrepancy.
Matou\v sek, Welzl, and Wernisch proved~\cite[Lemma~2.2]{matouvsek1993discrepancy}
that if $(U,\Ss)$ is a set system with $U\in\Ss$ and  $f$ is a
function such that $\disc(\Ss_{|X})\leq f(|X|)$ for all $X\subseteq U$, then
for every integer $t\geq 0$ there exists an
$\varepsilon$-approxi\-mation~$A$ for $\Ss$ with
$|A|=\bigl\lceil\frac{|U|}{2^t}\bigr\rceil$ and
$\varepsilon\leq \frac2{|U|}
\cdot \sum_{i=0}^{t-1}2^{i}f\bigl(\bigl\lceil\frac{|U|}{2^{i}}\bigr\rceil\bigr)
<2\,\herdisc(\Ss)\frac{2^t}{|U|}$.
%
In particular, every set system $(U,\Ss)$ has an $\varepsilon$-approximation 
of size at most $\bigl\lceil\frac{2\,\herdisc(\Ss)}{\varepsilon}\bigl\rceil$.
Furthermore, this construction is efficient and an $\epsilon$-approximation with the stated size can be computed in polynomial time.  Thus, as definable set systems can be constructed in polynomial time in bounded expansion classes (by using the model checking algorithm of \cite{DKT2}), we get the following. 

\begin{corollary}
	\label{cor:approx}
	For every class $\Cc$ with bounded expansion,  there exists an integer $c$ and an algorithm that computes, for a partitioned formula $\varphi(\bar x;\bar y)$, a graph $G\in\Cc$, and a real number and $\epsilon>0$, an $\epsilon$-approximation of $\Ss^\phi(G)$ in time $\Oof(|V(G)|^c)$, where the constant hidden in the $\Oof$-notation depends on $\Cc$, $\phi$, and $\epsilon$.
\end{corollary}

As the class of planar graphs has bounded expansion, we deduce  the following answer to \Cref{pb:discfo}.

\pbtwo*
\begin{solution}
	Consider the formula $\phi(\overline{x};\overline{y})$ with $\overline{x}=(x_1,x_2)$ and $\overline{y}=(y_1,y_2)$, which asserts that $x_1$ and $x_2$ are adjacent, and that there exists a path of length at most $100$ linking $y_1$ and $y_2$ including the edge $x_1x_2$.
	For any fixed bounded expansion class $\Cc$, like the class of planar graphs, there exists a constant $c$ such that for every $G\in\mathscr C$ we have
	$\herdisc(\Ss^\phi(G))\leq c$. It follows that the set system $(E(G),\Ss^\phi(G))$ has an $\epsilon$-approximation with size 
	$\bigl\lceil\frac{2c}{\varepsilon}\bigl\rceil$, which can be computed in polynomial time.
\end{solution}
\subsection{Clique coloring}
\label{sec:app3}
A \emph{clique coloring} of a graph $G$ is a coloring of the vertices of $G$ such that no maximal clique is monochromatic \cite{duffus1991two}. 
If a hereditary class of graphs has a bounded clique coloring number, it is obviously (exponentially) $\chi$-bounded. However,
while it was conjectured \cite{duffus1991two} that every perfect graph is $3$-colorable, Charbit et al.~\cite{charbit2016perfect} showed they do not have a bounded clique coloring number.

As a variant of clique coloring, one can ask whether a graph admits a vertex coloring by $k$ colors in such a way that no maximal clique $K$ contains more than $\lceil (1-\epsilon)\,|K|\rceil$ vertices\footnote{Note that the rounding allows to introduce a threshold, meaning that small cliques (of size at most $\lfloor 1/\epsilon\rfloor$) are allowed to be monochromatic.} of same color (for some fixed integer $k$ and positive real $\epsilon$). 
For instance, an easy induction shows that cographs do not have such a coloring in general, while classes with bounded shrubdepth (and, more generally, classes with structurally bounded expansion) do.

In this context, discrepancy colorings appear as a very strong property, by requiring that the vertices can be $2$-colored in such a way that in every maximal clique the difference between the numbers of vertices of each color is bounded by some constant $c$. Even classes with bounded shrubdepth do not have this property. However, we shall prove the following.

\begin{corollary}
	\label{thm:discK}
	For every class $\mathscr C$ with bounded expansion and every integer $k$ there is a constant $c$ such that for every graph $G\in\mathscr C$, the set system formed by all maximal cliques of~$G^k$ has hereditary discrepancy at most $c$.
\end{corollary}
\begin{proof}
	According to \cite{power}, there exists an integer $p$ (depending on $\mathscr C$ and $k$) such that for every $G\in\mathscr C$ and every maximal clique $K$ of $G$ there exist $v_1,\dots,v_p\in K$ with
	$K=\bigcap_{i=1}^p N_{G^k}[v_i]$. On the other hand, given $p$ vertices $u_1,\dots,u_p$ it is easily checked that the property that 
	$\bigcap_{i=1}^p N_{G^k}[u_i]$ induces a maximal clique of $G^k$ can be checked by a first-order sentence $\kappa(x_1,\dots,x_p)$. Hence the set system of all maximal cliques of $G^k$ (plus the empty set) is definable in $G$ as $\mathcal S^\phi(G)$, where $\phi(\overline{x};\overline{y}):= (\overline{x}=\overline{y})\wedge \kappa(\overline{y})$.
\end{proof}


\pbthree*
\begin{solution}
	
Let $G$ be a map graph, that is, the vertex-face incidence
graph of a planar map.
	By definition, $G$ is an induced subgraph of the square $H^2$ of a planar graph $H$.	According to \Cref{thm:discK}, the hereditary discrepancy of the set system of the maximal cliques of $H^2$ is bounded by some constant $c$. As every maximal clique of $G$ is the trace on $V(G)$ of some maximal clique of $H^2$ we deduce that the set system of all maximal cliques of $G$ has bounded discrepancy.
\end{solution}
\subsection{Quantifier elimination}
\label{sec:app4}
Our results on hereditary discrepancy unexpectedly allow us to answer \Cref{pb:elim}.

\pbfour*
\begin{corollary}
	\label{thm:noQE}
	Let $\Cc$ be the class of the $1$-subdivisions of all bipartite graphs, whose girth is greater than their maximum degree. The class $\mathscr C$ is nowhere dense and such that there does not exist an expansion $\sigma$ of the signature of graphs by unary relations and functions, an expansion $F$ of graphs into $\sigma$-structures and a quantifier-free formula $q(x,y)$ in the language of $\sigma$-structures with
	\begin{equation}
		\label{eq:genQE}
	\forall G\in\mathscr C\ \forall a,b\in V(G)\quad (G\models\phi(a,b))\iff (F(G)\models q(a,b)),
\end{equation}
	where $\phi(a,b)$ expresses that $a$ and $b$ are at distance $2$ in $G$.
\end{corollary}
\begin{proof}
	Let $\Cc_0$ be the class of all bipartite graphs whose girth is greater than the maximum degree. It is well known that this class is monotone, nowhere dense, but has unbounded degeneracy. Let $\Cc$ be the class of the $1$-subdivisions of the graphs in $\Cc_0$. This class is obviously also nowhere dense.
	
	Let $\sigma_0$ be the signature consisting of two unary functions $f_1$ and $f_2$.
	According to \Cref{obs:deg}, we can associate to each $G\in\Cc$ (which is the $1$-subdivision of a bipartite graph $H$ with parts $V_1$ and $V_2$) the $\sigma_0$-structure $\mathbf M$ defined as follows: if $v$ is a principal vertex (i.e.\ a vertex of~$H$) then $f_1(v)=f_2(v)=v$, while if $v$ is the subdivision vertex of the edge $\{u_1,u_2\}\in V_1\times V_2$ then $f_1(v)=u_1$ and $f_2(v)=u_2$. Let $\Dd$ be the class of all the $\sigma_0$-structures $\mathbf M$ obtained from the graphs $G\in\Cc$.

	Assume towards a contradiction that there is an expansion $\sigma$ of the signature of graphs by unary relations and functions, an expansion $F$ of graphs into $\sigma$-structures and a quantifier-free formula~$q(x,y)$ in the language of $\sigma$-structures such that equation \eqref{eq:genQE} holds.
	Then this defines an expansion $\sigma'$ of $\sigma_0$, an expansion $F'$ of $\sigma_0$ structures into $\sigma'$-structures, and a quantifier free formula $q'$, such that 
\[
	\forall \mathbf M\in\mathscr D\ \forall a,b\in M\quad (\mathbf M\models\phi'(a,b))\iff (F'(\mathbf M)\models q'(a,b)),
\]	
	where $\phi'(u,v)$ expresses that $u$ and $v$ are at distance $2$ in the Gaifman graph of $\mathbf M$.

\pagebreak
	According to \Cref{thm:QF} there is a constant $C$ such that 
	$\disc(\Ss^{q'}(F'(\mathbf M))\leq C$ for all $\mathbf M\in\Dd$.
	However, it is easily checked that if $\mathbf M$ is the $\sigma_0$-structure corresponding to the $1$-subdivision~$G$ of a graph $H$ we have
	$\Ss^E(H)\subseteq \Ss^\phi(G)=\Ss^{\phi'}(\mathbf M)$.
	According to \Cref{thm:BE_hdisc} the set systems $\Ss^E(H)$ with $H\in\Cc_0$ have unbounded discrepancy, what contradicts  $\disc(\Ss^{\phi'}(\mathbf M))=\disc(\Ss^{q'}(F'(\mathbf M))\leq C$.
\end{proof}

\section{Concluding Remarks}
\label{sec:conc}
As mentioned in \Cref{sec:setdef}, the set systems definable by a quantifier-free formula in pointer structures are basically the same as the neighborhood set systems of set-defined graphs.

A \emph{$k$-copy} of a graph $G$ is the graph obtained from $k$ copies of $G$ by adding an edge between any two clone vertices.
A \emph{simple (first-order) interpretation} $\mathsf I$ is a pair of formulas $(\eta(x,y),\nu(x))$; it maps a (possibly colored) graph $G$ into the graph $\mathsf I(G)$ with vertex set $\nu(G)=\{v\in G\colon G\models \nu(v)\}$ and edge set $\eta(G)\cap \nu(G)^2=\{\{u,v\}\in \nu(G)^2\colon G\models \eta(u,v)\vee \eta(v,u)\}$ (where the formulas~$\eta$ and~$\nu$ can use the colors).
A \emph{(first-order) transduction} $\mathsf T$ is a pair $(k,\mathsf I)$, where $k$ is an integer and $\mathsf I$ is a simple interpretation; it maps a graph $G$ to the class of all graphs that can be obtained from $G$ by computing a $k$-copy, then coloring this $k$-copy in an arbitrary way, then applying the transduction~$\mathsf I$. We refer the interested reader to \cite{SBE_TOCL, arboretum_arxiv} for a discussion on first-order transductions.

For a monotone class $\Cc$, it follows from \cite{Jiang2020} and \Cref{thm:BE_hdisc} that the four following properties are equivalent:
\begin{enumerate}
	\item $\Cc$ has bounded expansion;
	\item every transduction of $\Cc$ is set-defined;
	\item every transduction of $\Cc$ is linearly $\chi$-bounded;
	\item every set system definable in $\Cc$ has bounded hereditary discrepancy.
\end{enumerate}

The equivalence between the first three items is known to hold for biclique-free classes of graphs.
We conjecture that this extends to the equivalence with the last item.
\begin{conjecture}
	\label{conj:bic}
	A biclique-free class of graphs has bounded expansion if and only if every set system definable in $\Cc$ has bounded hereditary discrepancy.
\end{conjecture}

An intriguing related problem is the following, for which a positive answer would imply \cref{conj:bic}.
\begin{problem}
	\label{pb:T}
	Is it true that for every non-degenerate biclique-free class $\Cc$ of graphs there exists a transduction $\mathsf T$ such that
	$\mathsf T(\Cc)$ includes a monotone non-degenerate class of graphs?
\end{problem}

\newpage
\bibliographystyle{plain}
\bibliography{ref}

\begin{thebibliography}{10}

\bibitem{Adler2013}
Hans Adler and Isolde Adler.
\newblock Interpreting nowhere dense graph classes as a classical notion of
  model theory.
\newblock {\em European Journal of Combinatorics}, 36:322--330, 2014.

\bibitem{alon1999norm}
Noga Alon, Lajos R{\'o}nyai, and Tibor Szab{\'o}.
\newblock Norm-graphs: variations and applications.
\newblock {\em Journal of Combinatorial Theory, Series B}, 76(2):280--290,
  1999.

\bibitem{bansal2014algorithmic}
Nikhil Bansal.
\newblock Algorithmic aspects of combinatorial discrepancy.
\newblock In {\em A Panorama of Discrepancy Theory}, pages 425--457. Springer,
  2014.

\bibitem{BECK19811}
J\'ozsef Beck and Tibor Fiala.
\newblock {“Integer-making”} theorems.
\newblock {\em Discrete Applied Mathematics}, 3(1):1--8, 1981.

\bibitem{bednarchak1997note}
Debe Bednarchak and Martin Helm.
\newblock A note on the {Beck--Fiala} theorem.
\newblock {\em Combinatorica}, 17(1):147--149, 1997.

\bibitem{bukh2016improvement}
Boris Bukh.
\newblock An improvement of the {Beck--Fiala} theorem.
\newblock {\em Combinatorics, Probability and Computing}, 25(3):380--398, 2016.

\bibitem{charbit2016perfect}
Pierre Charbit, Irena Penev, St{\'e}phan Thomass{\'e}, and Nicolas Trotignon.
\newblock Perfect graphs of arbitrarily large clique-chromatic number.
\newblock {\em Journal of Combinatorial Theory, Series B}, 116:456--464, 2016.

\bibitem{charikar2011tight}
Moses Charikar, Alantha Newman, and Aleksandar Nikolov.
\newblock Tight hardness results for minimizing discrepancy.
\newblock In {\em Proceedings of the twenty-second annual ACM-SIAM symposium on
  Discrete Algorithms}, pages 1607--1614. SIAM, 2011.

\bibitem{chazelle2001discrepancy}
Bernard Chazelle.
\newblock {\em The discrepancy method: randomness and complexity}.
\newblock Cambridge University Press, 2001.

\bibitem{Cover2006}
T.M. Cover and J.A. Thomas.
\newblock {\em Elements of Information Theory}.
\newblock Wiley, 2nd edition, 2006.

\bibitem{diestel}
Reinhard Diestel.
\newblock {\em Graph Theory, 4th Edition}, volume 173 of {\em Graduate texts in
  mathematics}.
\newblock Springer, 2012.

\bibitem{drmota2006sequences}
Michael Drmota and Robert~F Tichy.
\newblock {\em Sequences, discrepancies and applications}.
\newblock Springer, 2006.

\bibitem{duffus1991two}
Dwight Duffus, Bill Sands, Norbert Sauer, and Robert~E Woodrow.
\newblock Two-colouring all two-element maximal antichains.
\newblock {\em Journal of Combinatorial Theory, Series A}, 57(1):109--116,
  1991.

\bibitem{Dvovrak2011}
Z.~Dvo{\v r}{\'a}k.
\newblock Constant-factor approximation of domination number in sparse graphs.
\newblock {\em European J. Combin.}, 34(5):833--840, 2013.

\bibitem{DKT2}
Zden{\v e}k Dvo{\v r}\'ak, Dan Kr{\'a}{\v l}, and Robin Thomas.
\newblock Testing first-order properties for subclasses of sparse graphs.
\newblock {\em Journal of the ACM}, 60:5 {\rm Article} 36, 2013.

\bibitem{SBE_TOCL}
J.~Gajarsk\'y, S.~Kreutzer, J.~Ne{\v s}et{\v r}il, P.~{Ossona de Mendez},
  M.~Pilipczuk, S.~Siebertz, and S.~Toru\'nczyk.
\newblock First-order interpretations of bounded expansion classes.
\newblock {\em ACM Transactions on Computational Logic}, 21(4):{Article 29},
  2020.

\bibitem{grohe2018coloring}
Martin Grohe, Stephan Kreutzer, Roman Rabinovich, Sebastian Siebertz, and
  Konstantinos Stavropoulos.
\newblock Coloring and covering nowhere dense graphs.
\newblock {\em SIAM Journal on Discrete Mathematics}, 32(4):2467--2481, 2018.

\bibitem{Grohe2013}
Martin Grohe, Stephan Kreutzer, and Sebastian Siebertz.
\newblock Deciding first-order properties of nowhere dense graphs.
\newblock In {\em Proceedings of the 46\textsuperscript{th} Annual ACM
  Symposium on Theory of Computing}, STOC '14, pages 89--98, New York, NY, USA,
  2014. ACM.

\bibitem{DBLP:journals/jacm/GroheKS17}
Martin Grohe, Stephan Kreutzer, and Sebastian Siebertz.
\newblock Deciding first-order properties of nowhere dense graphs.
\newblock {\em J. {ACM}}, 64(3):17:1--17:32, 2017.

\bibitem{Ha}
David Haussler.
\newblock Sphere pacjink numbers for subsets of the {B}oolean $n$-cube with
  bounded {Vapnik--Chervonenkis} dimension.
\newblock {\em Journal of Combinatorial Theory, Series A}, 69(2):217--232,
  1995.

\bibitem{hodges1997shorter}
Wilfrid Hodges et~al.
\newblock {\em A shorter model theory}.
\newblock Cambridge university press, 1997.

\bibitem{Jiang2020}
Yiting Jiang, Jaroslav Ne{\v s}et{\v r}il, Patrice {Ossona de Mendez}, and
  Sebastian Siebertz.
\newblock Regular partitions of gentle graphs.
\newblock {\em Acta Mathematica Hungarica}, 161(2):719--755, 2020.
\newblock Special issue dedicated to Endre Szemer\'edi's 80th birthday.

\bibitem{Kierstead2003}
Hal~A. Kierstead and Daqing Yang.
\newblock Orderings on graphs and game coloring number.
\newblock {\em Order}, 20:255--264, 2003.

\bibitem{levy2017deterministic}
Avi Levy, Harishchandra Ramadas, and Thomas Rothvoss.
\newblock Deterministic discrepancy minimization via the multiplicative weight
  update method.
\newblock In {\em International Conference on Integer Programming and
  Combinatorial Optimization}, pages 380--391. Springer, 2017.

\bibitem{libkin2013elements}
Leonid Libkin.
\newblock {\em Elements of finite model theory}.
\newblock Springer Science \& Business Media, 2013.

\bibitem{matouvsek1995tight}
Ji{\v{r}}{\'\i} Matou{\v{s}}ek.
\newblock Tight upper bounds for the discrepancy of half-spaces.
\newblock {\em Discrete \& Computational Geometry}, 13(3-4):593--601, 1995.

\bibitem{matouvsek1999combinatorial}
Ji{\v{r}}{\'\i} Matou{\v{s}}ek.
\newblock Combinatorial discrepancy.
\newblock In {\em Geometric Discrepancy}, pages 101--135. Springer, 1999.

\bibitem{matousek2009geometric}
Ji{\v{r}}{\'\i} Matou{\v{s}}ek.
\newblock {\em Geometric discrepancy: An illustrated guide}, volume~18.
\newblock Springer Science \& Business Media, 2009.

\bibitem{matouvsek1993discrepancy}
Ji{\v{r}}{\'\i} Matou{\v{s}}ek, Emo Welzl, and Lorenz Wernisch.
\newblock Discrepancy and $\epsilon$-approximations for bounded {VC}-dimension.
\newblock {\em Combinatorica}, 13(4):455--466, 1993.

\bibitem{power}
J.~Ne{\v s}et{\v r}il, P.~{Ossona de Mendez}, M.~Pilipczuk, and X.~Zhu.
\newblock Clustering powers of sparse graphs.
\newblock {\em Electronic Journal of Combinatorics}, 27(4):P4.17, 2020.

\bibitem{arboretum_arxiv}
J.~Ne{\v s}et{\v r}il, P.~Ossona~de Mendez, and S.~Siebertz.
\newblock Structural properties of the first-order transduction quasiorder.
\newblock arXiv:2010.02607 [math.CO], 2021.
\newblock accepted at CSL 2022.

\bibitem{Taxi_tdepth}
Jaroslav Ne{\v s}et{\v r}il and Patrice {Ossona de Mendez}.
\newblock Tree depth, subgraph coloring and homomorphism bounds.
\newblock {\em European Journal of Combinatorics}, 27(6):1022--1041, 2006.

\bibitem{POMNI}
Jaroslav Ne{\v s}et{\v r}il and Patrice {Ossona de Mendez}.
\newblock Grad and classes with bounded expansion {I}. {D}ecompositions.
\newblock {\em European Journal of Combinatorics}, 29(3):760--776, 2008.

\bibitem{POMNII}
Jaroslav Ne{\v s}et{\v r}il and Patrice {Ossona de Mendez}.
\newblock Grad and classes with bounded expansion {II}. {A}lgorithmic aspects.
\newblock {\em European Journal of Combinatorics}, 29(3):777--791, 2008.

\bibitem{Sparsity}
Jaroslav Ne{\v s}et{\v r}il and Patrice {Ossona de Mendez}.
\newblock {\em Sparsity (Graphs, Structures, and Algorithms)}, volume~28 of
  {\em Algorithms and Combinatorics}.
\newblock Springer, 2012.
\newblock 465 pages.

\bibitem{pilipczuk2018number}
Micha{\l} Pilipczuk, Sebastian Siebertz, and Szymon Toru{\'n}czyk.
\newblock On the number of types in sparse graphs.
\newblock In {\em Proceedings of the 33rd Annual ACM/IEEE Symposium on Logic in
  Computer Science}, pages 799--808. ACM, 2018.

\bibitem{sauer1972density}
Norbert Sauer.
\newblock On the density of families of sets.
\newblock {\em Journal of Combinatorial Theory, Series A}, 13(1):145--147,
  1972.

\bibitem{shelah1972combinatorial}
Saharon Shelah.
\newblock A combinatorial problem; stability and order for models and theories
  in infinitary languages.
\newblock {\em Pacific Journal of Mathematics}, 41(1):247--261, 1972.

\bibitem{shelah1990classification}
Saharon Shelah.
\newblock {\em Classification theory: and the number of non-isomorphic models}.
\newblock Elsevier, 1990.

\bibitem{siebertz2016nowhere}
S.~Siebertz.
\newblock {\em Nowhere Dense Classes of Graphs: Characterisations and
  Algorithmic Meta-Theorems}, volume~5.
\newblock Universit{\"a}tsverlag der TU Berlin, 2016.

\bibitem{spencer1985six}
Joel Spencer.
\newblock Six standard deviations suffice.
\newblock {\em Transactions of the American mathematical society},
  289(2):679--706, 1985.

\bibitem{Vapnik1971}
Vladimir~Naumovitch Vapnik and Alexe{\"\i}~Iakovlevitch {\v C}ervonenkis.
\newblock On the uniform convergence of relative sequences of events to their
  probabilities.
\newblock {\em Theory Probab. Appl.}, 16:264--280, 1971.

\bibitem{Zhu2008}
Xuding Zhu.
\newblock Colouring graphs with bounded generalized colouring number.
\newblock {\em Discrete Math.}, 309(18):5562--5568, 2009.

\end{thebibliography}
\end{document}